\documentclass[preprint,10pt,a4paper]{elsarticle2}

\usepackage{amsmath,amssymb,amsfonts,amsthm}
\usepackage{moreverb,rotating,graphics}
\usepackage[T1]{fontenc}
\usepackage[latin1]{inputenc}
\usepackage{epsfig}
\usepackage[english]{babel} 
\usepackage{color}
\usepackage[nottoc, notlof, notlot]{tocbibind}
\usepackage{dsfont}

\journal{ }

\newtheorem{theorem}{Theorem}[section]
\newtheorem{proposition}[theorem]{Proposition}

\newtheorem{lemma}[theorem]{Lemma}
\newtheorem{corollary}[theorem]{Corollary}
\newtheorem{definition}[theorem]{Definition}

\newtheorem{example}[theorem]{Example}

\renewcommand{\epsilon}{\varepsilon}
\newcommand{\trv}{{\rm \underline{Tr}\,}}
\newcommand{\tr}{{\rm Tr}\,}
\renewcommand{\P}{{\mathcal P}}
\newcommand\B{{\mathcal B}}
\renewcommand\S{{\mathfrak{S}}}

\newcommand\Sv{{\underline{\mathfrak{S}} }}

\newcommand\F{{\mathcal F}}

\renewcommand\H{{\mathcal H}}
\newcommand\E{{\mathcal E}}
\newcommand\K{{\mathcal K}}
\renewcommand\O{{\Omega}}
\renewcommand\o{{\omega}}
\newcommand\NN{{\mathbb N}}
\newcommand\ZZ{{\mathbb Z}}
\newcommand\RR{{\mathbb R}}
\newcommand\R{{\mathbb R}}
\newcommand\Z{{\mathbb Z}}

\newcommand\RRd{{\mathbb R ^d}}
\newcommand\ZZd{{\mathbb Z ^d}}
\newcommand\CC{{\mathbb C}}
\newcommand\PP{{\mathbb P}}
\newcommand\EE{{\mathbb E}}
\newcommand\ii{{\infty}}
\newcommand\1{{\mathds{1}}}
\newcommand\D{{\mathcal D}}

\newcommand\cn{{{n\rightarrow\ii}}}
\newcommand\cvL{{\operatorname*{\longrightarrow}_{L\rightarrow\ii}}}
\newcommand\cL{{{L\rightarrow\ii}}}

\newcommand{\norm}[1]{\left\| #1\right\|}
\newcommand{\set}[1]{\left\{ #1\right\}}
\newcommand{\bra}[1]{\left( #1\right)}
\newcommand{\av}[1]{\left| #1\right|}
\renewcommand{\phi}{\varphi}
\newcommand \dps{\displaystyle }
\newcommand{\tv}[1]{ {\rm \underline{Tr}}\,\left( #1\right)}
\newcommand{\Ev}[1]{ {\mathbb E}\left( #1\right)}
\newcommand{\wto}{\rightharpoonup}
\newcommand\pscal[1]{{\ensuremath{\left\langle #1 \right\rangle}}}

\begin{document}

\begin{frontmatter}
\title{Mean-field models for disordered crystals}
\author[Ponts]{\'Eric Cancès}
\author[Cergy,Ponts]{Salma Lahbabi}
\author[Cergy]{Mathieu Lewin}
\address[Ponts]{CERMICS, \'Ecole Nationale des Ponts et Chaussées (Paristech) \& INRIA (Micmac Project), 6-8 Av. Blaise Pascal, 77455 Champs-sur-Marne, France}
\address[Cergy]{CNRS \& Laboratoire de Mathématiques (UMR 8088), Université de Cergy-Pontoise, 95000 Cergy-Pontoise Cedex, France.}

\begin{abstract} 
In this article, we set up a functional setting for mean-field electronic structure models of Hartree-Fock or Kohn-Sham types for disordered crystals. The electrons are quantum particles and the nuclei are classical point-like particles whose positions and charges are random. We prove the existence of a minimizer of the energy per unit volume and the uniqueness of the ground state density of such disordered crystals, for the reduced Hartree-Fock model (rHF). We consider both (short-range) Yukawa and (long-range) Coulomb interactions. In the former case, we prove in addition that the rHF ground state density matrix satisfies a self-consistent equation, and that our model for disordered crystals is the thermodynamic limit of the supercell model.
\end{abstract}

\end{frontmatter}

\tableofcontents


\section{Introduction}
The modeling and simulation of the electronic structure of crystals is
one of the main challenges in solid state physics and materials science.
Indeed, a crystal contains an extremely large number (in fact an infinite number in
mathematical models) of quantum particles interacting through long-range
Coulomb forces. This complicates dramatically the mathematical analysis of such systems.

\medskip

Finite size molecular systems containing no heavy atoms can be
accurately described by the $N$-body Schr\"odinger equation, or its
relativistic corrections. Because of its very high complexity, this equation is often approximated by 
nonlinear models which are more amenable to numerical simulations. On the other
hand, no such reference model is available for infinite molecular
systems such as crystals. For this reason, in solid state
physics and material sciences, the electronic structure of crystals is often described 
by \emph{linear empirical models} on the one hand, and
\emph{mean-field models} of Hartree-Fock or Kohn-Sham types on the other hand.

In linear empirical models, the electrons in the crystal are seen as non-interacting
particles in an effective potential $V_{\rm eff}$, so that
their behavior is completely characterized by the effective Hamiltonian
$$
H = -\frac 12 \Delta + V_{\rm eff},
$$
a self-adjoint operator on $L^2(\RR^d)$. Here $d$ is the space dimension 
which is $d=3$ for usual crystals. The cases $d=1$ and $d=2$ are also of interest since
linear polymers and crystalline surfaces behave, in some respects, as
one- and two-dimensional systems, respectively. Throughout this article,
we adopt the system of atomic units in which $\hbar=1$, $m_e=1$, $e=1$
and $4\pi\epsilon_0=1$, where $\hbar$ is the reduced Planck constant,
$m_e$ the mass of the electron, $e$ the elementary charge, and
$\epsilon_0$ the dielectric permittivity of the vacuum. For the sake of
simplicity, we work with spinless electrons, but our arguments can be
straightforwardly extended to models with spin.

When the system under study is a perfect crystal, the effective
potential $V_{\rm eff}$ is an ${\cal R}$-periodic function $V_{\rm
per}$, where ${\cal R}$ is a discrete lattice of $\R^d$, and the
effective Hamiltonian is then a periodic Schr\"odinger operator on
$L^2(\RR^d)$,
$H=H_{\rm per} = -\frac{1}{2}\Delta + V_{\rm per}.$
The spectral properties of such operators are well-known~\cite{reed4}. Under some appropriate integrability conditions on $V_{\rm
per}$, it follows from Bloch theory that the spectrum of $H_{\rm
per}$ is purely absolutely continuous and composed of a countable number
of (possibly overlapping) bands.

It is possible to describe \emph{local defects} in such effective linear models. Displacing or changing the charge of a finite number of nuclei corresponds to adding a potential $W$ to $V_{\rm per}$. Because such perturbations are local, the potential $W$ decays at infinity and therefore the effective Hamiltonian 
$H_{\rm defect} = -\frac 12 \Delta + V_{\rm per}+W$
has the same essential spectrum as the unperturbed Hamiltonian $H_{\rm per}$. 
On the other hand, $H_{\rm defect}$ may possess discrete eigenvalues below its essential spectrum, or lying in spectral
gaps. They correspond to bound states of electrons in the presence of the local defects.

Doped semiconductors and alloys are examples of disordered crystals, which are perturbed in a non-local fashion. Such systems can be adequately modeled by random Schr\"odinger operators~\cite{carmona,caught}. One famous example is the
continuous Anderson model
$$
H_\omega = -\frac 12 \Delta + V_\omega \quad \mbox{with} \quad
V_\omega(x) = \sum_{k \in {\cal R}} q_k(\omega)\, \chi(x-k),
$$
where, typically, $\chi \in C^\infty_c(\R^d)$ and the $q_k$'s are
{i.i.d. random} variables. Here, only the charges are changed but it is possible to also account for stochastic displacements.
The study of the spectral properties of ergodic Schr\"odinger operators is a very active research topic (see e.g.
\cite{Hislop} and the references therein).

\medskip

In linear empirical models, the interactions between electrons are neglected (apart from the implicit
interaction originating from the Pauli principle preventing two
electrons from being in the same quantum state). Taking these interactions into account is however a necessity
for a proper physical description of these systems. One main difficulty is then that the Coulomb interaction 
is long-range and screening becomes extremely important to explain the macroscopic stability of such systems. Understanding screening effects in a precise manner is a difficult mathematical question.

As already mentioned above, there is no well-defined many-body Schr\"odinger equation for crystals. The
only available way to rigorously derive models for interacting electrons
in crystals is to use a thermodynamic limit procedure. The idea is to confine the system to a box, with suitable boundary conditions, and to study the limit when the size of the box grows to infinity. For stochastic many-body systems based on Schrödinger's equation, it is sometimes possible to show that the limit exists. In~\cite{Veniaminov}, Veniaminov has first considered a many-body quantum system with short range interactions. Short after, the existence of the limit for a crystal made of quantum electrons and stochastic nuclei interacting through Coulomb forces was shown in~\cite{BlaLew-12}, by Blanc and the third author of this article. In these two works dealing with the true many-body Schrödinger equation, the value of the thermodynamic limit is not known. For Thomas-Fermi-type models, Blanc, Le Bris and Lions were able to identify the thermodynamic limit and to study its properties~\cite{BLBL2007}. Unfortunately, Thomas-Fermi theory is not able to reproduce important physical properties of stochastic quantum crystals, like the Anderson localization under weak disorder.

The purpose of the present work is to propose and study a \emph{mean-field} (Hartree-Fock type) model which can be obtained from a thermodynamic limit procedure, for an infinite, randomly perturbed, interacting quantum crystal. This model is not as precise as the many-body Schrödinger equation, but it is still much richer than Thomas-Fermi type theories. In particular, it seems adequate for the description of Anderson localization in infinite interacting systems.

\medskip

More specifically, we consider a random nuclear charge $\mu(\omega,x)\geq0$. For simplicity we do not consider point-like charges, and we assume that $\mu(\omega,\cdot)\in L^1_{\rm loc}(\R^d)$ almost surely. Also we are interested in describing random perturbations which have some space invariance, and we make the assumption that they are the same when the system is translated by any vector of the underlying periodic lattice $\cal R$. We assume that the group $\cal R$ acts on the probability space in an ergodic fashion and we always make the assumption that $\mu$ is \emph{stationary}, which means $\mu(\tau_k(\omega),x)=\mu(\omega,x+k)$, where $\tau=(\tau_k)_{k\in\cal R}$ is the ergodic group action on the probability space. A typical example is given by a lattice $\cal R$ with one nucleus per unit cell, whose charge and position are perturbed by i.i.d. random variables,
$$\mu(\omega,x)=\sum_{k\in\cal R}q_k(\o)\,\chi\big(x-k-\eta_k(\o)\big).$$
Similarly, the state of the electrons in the crystal is modelled by a \emph{one-particle density matrix}~\cite{stability}, that is, a random family of operators $\gamma(\omega):L^2(\R^3)\to L^2(\R^3)$ such that $0\leq\gamma(\omega)\leq1$ almost surely. It is also assumed that $\gamma$ is stationary in the sense that its kernel satisfies $\gamma(\tau_k(\omega),x,y)=\gamma(\omega,x+k,y+k)$ for all $k\in\cal R$. These concepts will be recalled later in Section~\ref{sec:basics}.

For any such electronic state $\gamma$ we define in Section~\ref{sec:rHF} the corresponding reduced-Hartree-Fock (rHF) energy, in the field induced by the nuclear charge $\mu$. This energy is just the sum of the kinetic energy per unit volume of $\gamma$ and the potential energy per unit volume of $\gamma$ and $\mu$. The rHF model is obtained from the generalized Hartree-Fock model~\cite{LieSim-77,BacLieSol-94} by removing the exchange term~\cite{Solovej}. Alternatively, it can be seen as an extended Kohn-Sham model~\cite{DreizlerGross} with no exchange-correlation.

Defining the rHF energy properly requires to introduce several tools, which is the purpose of Sections~\ref{sec:electronic_states} and~\ref{sec:Coulomb-Yukawa}. We start by defining the average number of particles and the kinetic energy per unit volume for ergodic density matrices and we show useful inequalities. In particular we derive Hoffmann-Ostenhof~\cite{OO} and Lieb-Thirring inequalities~\cite{LiTh,LiTh2} for ergodic density matrices, which are very important estimates that we use all the time. Loosely speaking, they can respectively be stated as follows:
$$\text{Average kinetic energy per unit vol. of $\gamma$}\ \geq\ \EE\bra{\int_{Q}|\nabla\sqrt{\rho_\gamma}|^2}$$
and
$$\text{Average kinetic energy per unit vol. of $\gamma$}\ \geq\ K\,\EE\bra{\int_{Q}\rho_\gamma^{5/3}},$$
where $Q$ is the unit cell, $\rho_\gamma$ is the density of the state $\gamma$ and $K$ is a constant independent of $\gamma$. 

In Section~\ref{sec:Coulomb-Yukawa}, we discuss Poisson's equation
\begin{equation}
-\Delta V=4\pi \rho
\label{eq:Poisson-intro} 
\end{equation}
for stationary functions $\rho(\omega,x)$, where $\Delta$ is the Laplace operator with respect to the $x$-variable, and we explain that the situation is much more complicated than in the periodic case. In particular, the neutrality condition $\EE(\int_Q\rho)=0$
on the charge density appearing on the right side of \eqref{eq:Poisson-intro} is necessary but in general not sufficient to find a stationary solution $V$. When $\EE(\int_Q\rho^2)<\ii$ and $\EE(\int_Q\rho)=0$, it is possible to give a necessary and sufficient condition for the existence of a stationary solution $V$ to~\eqref{eq:Poisson-intro} such that $\EE(\int_QV^2)<\ii$. In words, $\rho$ should be in the range of the ``stationary Laplacian'' which is a particular self-adjoint extension of $-\Delta$ on $L^2(\Omega\times Q)$ with ``stationary boundary conditions''. 

Understanding Poisson's equation~\eqref{eq:Poisson-intro} for general stochastic charge densities $\rho$ is an important and interesting problem in itself. In order to define the associated Coulomb energy per unit volume, we adopt here a simple strategy and take the limit $m\to0$ of the Yukawa energy. This means we consider the regularized equation
\begin{equation}
-\Delta V_m+m^2V_m=4\pi \rho
\label{eq:Yukawa-intro} 
\end{equation}
and we define the Coulomb energy as the limit of $\EE\bra{\int_Q V_m\rho}$ when $m\to0$. We then give in Section~\ref{sec:Coulomb-Yukawa} several properties of this energy.

After these preliminaries, we are able to properly define and study the reduced Hartree-Fock energy for stochastic crystals in Section~\ref{sec:rHF}. In particular we prove the existence of a minimizer $\gamma$ of this energy and the uniqueness of the ground state density $\rho_\gamma$. In the Yukawa case $m>0$, we also show that the minimizers solve a self-consistent equation of the form
\begin{equation}
\begin{cases}
\displaystyle\gamma=\1_{(-\ii,\epsilon_{\rm F})}\left(-\frac{1}{2}\Delta+V_m\right),\\[0,3cm]
-\Delta V_m+m^2V_m=4\pi \big(\rho_\gamma-\mu\big).
\end{cases} 
\label{eq:SCF_intro}
\end{equation}
The mean-field operator
$$H_m=-\frac{1}{2}\Delta+V_m$$
is a random Schrödinger operator describing the collective behavior of the electrons in the system. Studying its spectral properties would allow to understand localization properties in the interacting stochastic crystal. 

In Section~\ref{limit_thermo}, we finally prove that, in the Yukawa case, our model is actually the thermodynamic limit of the supercell reduced Hartree-Fock theory (the system is confined to a box with periodic boundary conditions). This justifies our theory with Yukawa interactions. For Coulomb forces, our proof does not apply because of some missing screening estimates. We make more comments about this later in Section~\ref{limit_thermo}.

Let us end this introduction by mentioning that our theory is rather general and it actually works for any reasonable interaction potential which decays fast enough at infinity. We concentrate on the Yukawa interaction because of the limit $m\to0$ which corresponds to the more physical Coulomb case and which we study as well in this paper. Note that we consider here the action of a discrete group on $\O$ because we have in mind the case of a randomly perturbed crystal. Our approach can also be applied to the case when the group acting on $\O$ is $\RRd$ (amorphous material). We refer to~\cite{these} for details.

\bigskip

\noindent\textbf{Acknowledgement.}
The research leading to these results has received funding from the European Research Council under the European Community's Seventh Framework Programme (FP7/2007--2013 Grant Agreement MNIQS no. 258023).

\section{Electronic states in disordered crystals}\label{sec:electronic_states}

In mean-field models (such as Hartree-Fock or Kohn-Sham), the state of the electrons is described by a self-adjoint operator $\gamma$ acting on $L^2(\R^3)$, satisfying $0\leq\gamma\leq1$ in the sense of quadratic forms, and such that $\tr(\gamma)$ is the total number of electrons in the system~\cite{stability}. In (infinite) crystals, we always have $\tr(\gamma)=+\ii$. Such an operator $\gamma$ is called a \emph{(one-particle) density matrix}. The purpose of this section is to recall the main properties of  electronic states in a class of \emph{random} media, satisfying an appropriate invariance property called \emph{stationarity}.

\subsection{Basic definitions and properties}\label{sec:basics}

Throughout this paper, $d$ will denote the space dimension. We will later focus on the cases where  $d\in\{1,2,3\}$, but we keep $d$ arbitrary in this section. We restrict ourselves to the cubic lattice group $\mathcal{R}=\Z^d$ to simplify the notations; general discrete subgroups $\cal R$ can be tackled similarly without any additional difficulty. We consider a probability space $(\Omega, \mathcal{F},\mathbb P)$ and an ergodic group action $\tau$ of $\ZZd$ on $\O$. We recall that $\tau$ is called ergodic if it is measure preserving and if for any $A\in\F$ satisfying $\tau_k(A)=A$ for all $k\in \ZZd$, it holds that $\mathbb P(A)\in \{0,1\}$. 

\begin{example}[i.i.d. charges]\label{example1}
A typical probability space we have in mind is the one arising from a random distribution of particles of charges $q_1$ and $q_2$ on the sites of the lattice $\ZZd$ with probabilities $p_1$ and $1-p_1$. The probability space is then given by $\Omega=\{q_1,q_2\}^{\mathbb{Z}^d}$ and $\mathbb P =p^{\otimes\ZZd}$ where $p=p_1\delta_{q_1}+ (1-p_{1})\delta_{q_2}$. In this case, the group action is $\tau_k(\omega)=\omega_{\cdot+k}$.
\end{example}

The ergodic theorem \cite[Theorem 6.1, Theorem 6.4]{tempelman}, which will be extensively used in the sequel, can be stated as follows: 
\begin{theorem}[Ergodic theorem]
If $\tau$ is an ergodic group action of $\ZZd$ on $\O$ and $X\in L^p(\O)$, with $1 \leq p<\ii$, then, 
$$
\lim_{n\rightarrow \ii}\frac{1}{\bra{2n+1}^d}\sum_{k\in\ZZd\cap \left[-\frac{n}2,\frac{n}2\right]}X(\tau_k(\o))=\EE(X),$$
almost surely and in $L^p\bra{\O}$.
\end{theorem}

\medskip

A measurable function $f:\Omega \times \mathbb R ^d \rightarrow \mathbb C$ is called \emph{stationary}  if
$$\forall k\in \ZZd,\; f(\tau_k(\omega),x)=f(\omega, k+x), \;  \mbox{a.s.} \; \mbox{and}\;  \mbox{a.e. } $$
We will make use of the families of stationary function spaces
$$
L^p_s\bra{L^q}=\left\{ f \in L^p\left(\O,L^q_{\rm loc}\left( \RRd\right)\right) \; |\;  f \;\mbox{is stationary}\right\},
$$
and 
$$
H^m_s=\left\{ f \in L^2\left(\Omega, H^m_{\rm loc}\left(\RRd\right)\right) \; |\;  f \; \mbox{is stationary}\right\},
$$
and resort, for convenience, to the shorthand notation $L^p_s=L^p_s\bra{L^p}$.
Endowed with the norms 
$$
\|f\|_{L^p_s\bra{L^q}}=\mathbb E \left(\|f\|^p_{L^q(Q)}\right)^\frac{1}{p},
$$
and the scalar products 
$$
\langle f,g\rangle_{L^2_s}=\mathbb E \left(\langle f,g\rangle_{L^2(Q)}\right),
\qquad
\langle f,g\rangle_{H^m_s}=\mathbb E \left(\langle f,g\rangle_{H^m(Q)}\right),
$$
where 
$$Q:=\left[-\frac{1}{2},\frac{1}{2}\right)^d$$ 
denotes the semi-open unit cube, the spaces $L^p_s\bra{L^q}$ are Banach spaces and the spaces $L^2_s$ and $H^m_s$ are Hilbert spaces.

\medskip

We denote by $\H=L^2\bra{\RRd}$ the space of complex valued, square integrable functions, equipped with its usual scalar product $\langle\cdot,\cdot \rangle_{L^2}$. We also denote by
\begin{itemize}
\item $\B$ the space of the bounded linear operators on $\H$, endowed with the operator norm $\norm{\cdot}$;
\item ${\mathcal S}$ the space of the bounded self-adjoint operators on $\H$;
\item $\S_p$ the $p^{\rm th}$ Schatten class on $\H$. Recall that $\S_1$ is the space of the trace class operators on $\H$ and $\S_2$ the space of the Hilbert-Schmidt operators on $\H$.
\end{itemize}
Let $\mathcal D$ be a dense linear subspace of $\H$. A \emph{random operator} with domain $\mathcal D$ is a map $A$ from $\O$ into the set of the linear operators on $\H$ such that $\mathcal D\subset D (A(\o))$ a.s. and such that the map $\o\mapsto \langle A(\o)x,y\rangle$ is measurable for all $x\in \mathcal D$ and $y\in\H$. 

Of importance to us will be the \emph{uniformly bounded random operators} $A$ which are such that $\sup\mbox{ess}_{\o\in\O}\norm{A(\omega)}<\infty$. The Banach space of such operators is denoted by $L^\ii(\Omega,\B)$. This is a $W^*$-algebra which is known to be the dual of $L^1(\Omega,\S_1)$ (see, e.g.,~\cite[Corollary 3.2.2]{Sakai}). We will often use the corresponding weak-$\ast$ topology on $L^\ii(\Omega,\B)$ for which $A_n\wto_\ast A$ means
$$\EE\bra{\tr(A_nB)}\to\EE\bra{\tr(AB)}$$
for all $B\in L^1(\Omega,\S_1)$. Since $L^1(\Omega,\S_1)$ is separable, any bounded sequence $(A_n)$ in $L^\ii(\Omega,\B)$ has a subsequence $(A_{n_k})$ which converges weakly-$\ast$ to some $A\in L^\ii(\Omega,\B)$. Similarly, we know that the dual of $L^p(\Omega,\S_{q})$ is nothing else but $L^{p'}(\Omega,\S_{q'})$ when $1=1/p+1/p'=1/q+1/q'$ and $1\leq p,q<\ii$.

\medskip

Let $(U_k)_{k\in \ZZd }$ be the group of unitary operators on $\H$ defined by 
$$
U_kf(x)=f(x+k),\;\mbox{a.e.},\; \forall f\in \H,\; \forall k\in \ZZd.
$$
A random operator $A$ (not necessarily uniformly bounded) is called \emph{ergodic} or \emph{stationary} if for any $k\in \ZZd$, $\mathcal D\subset U_k(\mathcal D)$ and the following equality holds
$$
A(\tau_k(\o))=U_kA(\o)U_k^*,\;  \mbox{a.s.}
$$
One of the fundamental theorems for ergodic operators \cite[Theorem 1.2.5 p.13]{caught} states that for 
any self-adjoint ergodic operator $A$, there exists a closed set $\Sigma \subset \RR$ and a set $\O_1\in\F$ with $\PP(\O_1)=1$, such that
$\sigma(A(\o))=\Sigma$, for all $\o\in\O_1$. The set $\Sigma$ is called the almost sure spectrum of $A$.

\medskip

We finally denote by  $\underline{\mathcal S }$ the space of the ergodic operators on $\H$ that are almost surely bounded and self-adjoint.

\subsection{Ergodic locally trace class operators}

In this section, we recall the definitions of the trace per unit volume, the density and the kernel of an ergodic locally trace class operator (see e.g.~\cite{bouclet,these_nicolas}). For $1\leq p \leq \ii$, we denote by $L^p_c\bra{\RRd}$ the space of the compactly supported $L^p$ functions on $\RRd$.

\begin{definition}[Locally trace-class operators]
A random operator $A$ is called \emph{locally trace class} if $\chi A\chi\in L^1(\Omega,\S_1)$ for all $\chi\in L^\infty_c(\RRd)$, that is,
$$
\forall\chi\in L^\ii_c(\RRd),\quad\EE\Big(\tr\bra{\big|\chi A(\cdot)\chi\big|}\Big)<\infty.
$$
\end{definition}

We now focus on the particular case of ergodic operators, and denote by $\Sv_1$ the space of the ergodic, locally trace class operators. The following characterization of the positive operators of $\Sv_1$ will be useful.

\begin{proposition}[Characterization of ergodic locally trace-class operators]\label{magie}
Let $A$ be a positive, almost surely bounded, ergodic operator. Then  $A$ is locally trace class if and only if 
$\EE(\tr(\1_QA(\cdot)1_Q))<\infty$. 
\end{proposition}

The trace per unit volume of an operator $A \in \Sv_1$ is defined as
\begin{equation}\label{eq:trace-per-unit-volume}
\boxed{\tv{A}=\EE\big(\tr\left(\1_QA\left(\cdot\right)1_Q\right)\big).}
\end{equation}

The following summarizes the main properties of locally trace-class ergodic operators.

\begin{proposition}[Kernel and density]\label{kernel_er}
Let $A \in \Sv_1$. Then, there exists a unique function $A(\cdot,\cdot,\cdot)\in L^1(\O, L^2_{\rm loc}(\RRd\times \RRd))$, called the \emph{kernel} of $A$, and 
a unique function $\rho_A\in L^1_s$, called the \emph{density} of $A$, such that
$$
\forall \phi\in L^2_c(\RRd) ,\;  \bra{A(\o)\phi}(x)=\int_\RRd A(\o,x,y)\phi(y)\,dy\; \mbox{ a.s. and a.e.}
$$
and
\begin{equation}\label{eq:def_rho_A}
\forall \chi\in L^\infty_c(\RRd),\;\tr(\chi A(\o) \chi)=\int_\RRd \chi^2(x)\rho_A(\o,x)\,dx \;\mbox{ a.s.}
\end{equation}
The kernel $A(\cdot,\cdot,\cdot)$ is stationary in the following sense
$$
A(\tau_k(\o),x,y)=A(\o,x+k,y+k),\; \forall k\in \ZZd\;  \mbox{ a.e. and a.s. }
$$
Moreover, if $A\geq 0$, then $\rho_A\geq 0$.
\end{proposition} 

Note that it follows from (\ref{eq:trace-per-unit-volume}) and (\ref{eq:def_rho_A}) that
$$
\tv{A}=\EE\left(\int_Q \rho_A\right).
$$
The proofs of Propositions~\ref{magie} and~\ref{kernel_er} are elementary; they can be read in~\cite{these}.  

\medskip

The following cyclicity property is proved in~\cite{these_nicolas}, based on arguments in \cite{dixmier} (see also~\cite{these} for a detailed proof): if $B$ is an ergodic operator in $L^\infty\left(\O,\B\right)$ and $A$ a positive operator in $\Sv_{1}\cap L^\infty(\O,\B)$, then $BA$ and $AB$ are in $\Sv_{1}\cap L^\infty(\O,\B)$ and
\begin{equation}\label{cyclicity}
\trv\left( BA\right)=\trv\left( AB\right).
\end{equation}

\subsection{Ergodic operators with locally finite kinetic energy}

Ergodic density matrices for fermions are operators $\gamma\in\Sv_{1}\cap \underline{\mathcal S}$ such that $0\leq \gamma\leq 1$ a.s. By Birkhoff's theorem, the trace per unit volume can be interpreted from a physical viewpoint as the average number of particles per unit volume. In this section, we define and study in a similar fashion the average kinetic energy per unit volume.

\subsubsection{Definition}

For $1\leq j \leq d$, as usual, we denote by $P_j=-i\partial_{x_j}$ the momentum operator in the $j^{\rm th}$ direction, which is self-adjoint with $D(P_j)=\{\phi\in\H\ |\ \partial_{x_j}\phi\in\H\}$. As $P_j$ commutes with the translations, we see that for all $A\in \Sv_1$, the operator $P_j A P_j$ is ergodic. The operator $P_j A P_j$ is well defined and bounded on $D(P_j)$, with values in $D(P_j)'$, where $D(P_j)'$ is the topological dual space of $D(P_j)$. We say that the kinetic energy of $A$ is locally finite if $P_j A P_j\in\Sv_1$, and we then call 
$$\boxed{\tv{-{\Delta} A}:=  \sum_{j=1}^d\tv{P_j A P_j}}$$ 
the average kinetic energy per unit volume of $A$. We denote by $\Sv_{1,1}$ the subspace of $\Sv_1$ composed of the ergodic locally trace class operators with locally finite kinetic energy. 

\subsubsection{Hoffmann-Ostenhof and Lieb-Thirring inequalities for ergodic operators}

For finite systems ($\gamma\in\S_1\cap \mathcal{S}$, $0\leq \gamma\leq 1$ and $\tr\bra{-\Delta\gamma}<\ii$), the Hoffmann-Ostenhof~\cite{OO,stability} and Lieb-Thirring~\cite{LiTh,LiTh2,stability} inequalities provide useful properties of the map $\gamma \mapsto \rho_\gamma$. In this section, we state and prove an equivalent of these inequalities for ergodic density matrices with locally finite kinetic energy.

\begin{proposition}[Hoffmann-Ostenhof inequality for ergodic operators]\label{A2}
 Let $A$ be a positive operator in $\Sv_{1,1} \cap \underline{\mathcal S }$. Then 
$$\sqrt{\rho_A}\in H_s^1  \quad \mbox{and}\quad  \EE\left( \int_Q\left|\nabla\sqrt{\rho_A}\right|^2\right)\leq  \trv(-{\Delta} A).$$
\end{proposition}

\begin{proof}
It follows from Proposition \ref{kernel_er} that $\sqrt{\rho_A}\in L^2_s$. 
Let $B$ be a compact set of $\RRd$ and $\eta\in C^\infty_c(\RRd)$ such that $\eta \equiv 1$ on $B$ and $0\leq \eta\leq 1$. The operator $\eta A(\o)\eta$ has finite kinetic energy a.s. Therefore, the Hoffmann-Ostenhof inequality gives 
$$
\av{\nabla \sqrt{\rho_{A\bra{\o}}}}=|\nabla\sqrt{\rho_{\eta A(\o)\eta}}|\leq \sqrt{\sum_{n\in\NN}\lambda_n(\o)|\nabla\phi_n(\o)|^2} \mbox{ a.s. and a.e. on } B,
$$
where $\bra{\phi_n(\o)}_{n\in\NN}$ is an orthonormal basis of eigenvectors of the compact self-adjoint operator $\eta A(\o)\eta$ and  $\bra{\lambda_n(\o)}_{n\in\NN}$ the associated eigenvalues.  As 
$$\int_B \sum_{n\in\NN}\lambda_n(\o)|\nabla\phi_n(\o)|^2=\sum_{j=1}^d\tr\left( \1_B P_j \eta A\bra{\o} \eta P_j \1_B \right) \mbox{ a.s.}$$
and as for all $1\leq j\leq d$, $\1_B\left[P_j,\eta\right]=-i\1_B\partial_{x_j}\eta=0$, we deduce that
$$
\int_B|\nabla\sqrt{\rho_{A(\o)}}|^2 \leq \sum_{j=1}^d\tr\left( \1_B\eta  P_j  A\bra{\o} P_j \eta  \1_B \right)=\sum_{j=1}^d\tr\left( \1_B  P_j  A\bra{\o} P_j  \1_B \right).
$$
Therefore
$$
\EE\left(\int_B|\nabla\sqrt{\rho_{A}}|^2 \right)\leq \sum_{j=1}^d\Ev{\tr\left( \1_B  P_j  A P_j  \1_B \right)}.
$$
As $A$ has locally finite kinetic energy, we conclude that $\sqrt{\rho_A}\in H^1_s$. For $B=Q$, we obtain the stated inequality.
\end{proof}

The following corollary is an obvious consequence of Proposition~\ref{A2} and of the Sobolev embeddings.

\begin{corollary}\label{espace_f}
Let $A$ be a positive operator in $\Sv_{1,1} \cap \underline{\mathcal S }$. Then, $\rho_A\in L^1_s\bra{L^p}$, for $p=+\infty$ if $d=1$,   $p\in \left[ 1,+\ii\right)$ if $d=2$ and 
$1\leq p\leq \frac{d}{d-2}$ if $d\geq3$.
\end{corollary}

The following is now the ergodic equivalent of the Lieb-Thirring inequality~\cite{LiTh,LiTh2,stability}.

\begin{proposition}[Lieb-Thirring inequality for ergodic operators]\label{Lieb_T}
There exists a constant $K(d)>0$, depending only on the space dimension $d\geq1$, such that for all $\gamma\in\Sv_{1,1} \cap \underline{\mathcal S }$ with $0\leq \gamma(\o) \leq 1$ a.s., 
\begin{equation}\label{eq:LT}\rho_\gamma\in L_s^{\frac{d+2}{d}}\quad \mbox{and} \quad  K(d)\;\EE\left( \int_Q \rho_\gamma^{\frac{d+2}{d}} \right)\leq  \trv(-{\Delta} \gamma).
 \end{equation}
\end{proposition}

\begin{proof}
To prove~\eqref{eq:LT}, we apply the Lieb-Thirring inequality in a box of side-length  $L$, and then let $L$ go to infinity. The constant $K(d)$ can be chosen equal to the optimal Lieb-Thirring constant in the whole space.
Let $\Gamma_L=\left[ -{L}/{2},{L}/{2} \right)^d$ and let $\bra{\chi_L}_{L\in\NN^*}$ be a sequence of localizing functions in $C^\infty_c\bra{\RRd}$, such that $0\leq\chi_L\leq 1$, $\chi_L\equiv1$ on $\Gamma_{L-1}$, $\chi_L\equiv0$ outside of $\Gamma_L$, and $|\nabla\chi_L(x)|\leq C$. We first apply the Lieb-Thirring inequality to $\chi_L \gamma (\o)\chi_L$ and  obtain
$$
K(d)\int_{\Gamma_{L-1}}\rho_\gamma (\o,x)^{\frac{d+2}{d}}\,dx\leq \tr\left(-{\Delta} \chi_L \gamma (\o)\chi_L\right)\;\mbox{ a.s.}
$$
Next, using the stationarity of $\rho_\gamma$ and the equality $\left[P_j,\chi_L\right]=-i\partial_{x_j} \chi_L$, we get for any $\epsilon>0$
\begin{multline}\label{LT1}
K(d)\;\Ev{ \int_Q\rho_\gamma^{\frac{d+2}{d}}}\leq \frac{\bra{1+\epsilon}}{\bra{L-1}^d}\sum_{j=1}^d \Ev{ \tr\bra{\chi_L P_j \gamma P_j \chi_L}}\\
 + \frac{1+1/\epsilon}{\bra{L-1}^d}\sum_{j=1}^d \Ev{ \tr\bra{ \bra{\partial_{x_j}\chi_L} \gamma \bra{\partial_{x_j}\chi_L}}}.
\end{multline}
For each $1\leq j\leq d$, we have
$$
\tr\left(\chi_L P_j\gamma(\omega) P_j\chi_L\right) \leq \sum_{k\in \Gamma_L\cap \ZZd}  \tr\left( \1_Q P_j\gamma(\tau_k\bra{\omega})P_j \1_Q\right) \; \mbox{ a.s.}
$$
It follows from the ergodicity of $\gamma$ (hence of $P_j\gamma P_j$) that
\begin{equation}\label{LT2}
\sum_{j=1}^d \EE\left( \tr\left(\chi_L P_j \gamma P_j\chi_L\right)  \right) \leq L^d\tv{-{\Delta} \gamma }.
\end{equation}
Besides,
\begin{equation*}
 \tr\bra{(\partial_{x_j}\chi_L) \gamma (\partial_{x_j}\chi_L)} =\int_{\Gamma_L\setminus\Gamma_{L-1}}\rho_{\gamma }(\o,\cdot)(\partial_{x_j}\chi_L)^2
\leq C \int_{\Gamma_L\setminus\Gamma_{L-1}}\rho_{\gamma }(\o,\cdot),\\
\end{equation*}
where we have used that $\nabla \chi_L$ is uniformly bounded. Using again the stationarity of $\rho_\gamma $, we obtain
\begin{equation}\label{LT3}
\sum_{j=1}^d \EE\left(\tr\bra{\partial_{x_j}\chi_L} \gamma \bra{\partial_{x_j}\chi_L} \right)\leq CL^{d-1}\tv{\gamma}.
\end{equation}
Combining~\eqref{LT1},~\eqref{LT2} and~\eqref{LT3}, letting $L$ go to infinity then letting $\epsilon$ go to $0$, we end up with the claimed inequality.
\end{proof}

\subsubsection{A compactness result}

In this section we investigate the weak compactness properties of the set of fermionic density matrices with finite average number of particles and kinetic energy per unit volume
\begin{equation}
\K:=\left\{ \gamma\in \Sv_{1,1}\cap \underline{\mathcal S},\;  0\leq\gamma\leq 1 \mbox{ a.s.}\right\}.
\label{def:K}
\end{equation}
This set is a weakly-$\ast$ closed convex subset of $L^\ii(\Omega,\B)$. The following result will be very useful.

\begin{proposition}[Weak compactness of ergodic density matrices]\label{prop:weak-compactness}
Let $(\gamma_n)$ be any sequence in $\K$. Then there exists $\gamma\in\K$ and a subsequence $(n_k)$ such that 
\begin{enumerate}
\item \label{item10}$\dps \gamma_{n_k} \mathop{\rightharpoonup_*}_{k \to \infty} \gamma$  in $L^\ii(\Omega,\B)$,

\smallskip

\item \label{item11} $\dps \lim_{k\to\ii}\tv{\gamma_{n_k}}=\tv{\gamma}$,

\smallskip

\item \label{item12} $\dps \rho_{\gamma_{n_k}}\mathop{\rightharpoonup}_{k \to \infty} \rho_\gamma\; \mbox{weakly in}\; L^{\frac{d+2}{d}}_s$,

\smallskip

\item \label{item13} $\dps \trv\left(-{\Delta} \gamma\right)\leq \liminf_\cn \trv\left(-{\Delta} \gamma_{n}\right)$.
\end{enumerate}
\end{proposition}

Note that, in average, there is never any loss of particles when passing to weak limits: $\tv{\gamma_n}$ tends to $\tv{\gamma}$ as $n\to\ii$. On the other hand, even if we have $\rho_{\gamma_n}\rightharpoonup\rho_\gamma$ weakly and $\EE(\int_Q\rho_{\gamma_n})\to\EE(\int_Q\rho_\gamma)$, in general we do not have almost sure convergence and we do not expect strong convergence in $L^p_s$ for $1\leq p\leq 1+2/d$. 

\begin{example}[Weak versus strong convergence for $\rho_{\gamma_n}$]
Consider a smooth function $\phi$ with compact support in the ball $B(0,1/2)$ such that $\norm{\phi}_{L^2}=1$, and the operator
$$\gamma_n=\sum_{k\in\Z^d}\frac{1+\sin(2\pi n \omega_k)}2\; \big|\phi(\cdot+k)\big\rangle\big\langle\phi(\cdot+k)\big|.$$
where $\bra{\o_k}_{k\in\ZZd}$ are i.i.d. variables, uniformly distributed on $[0,1]$.
Then we have $\gamma_n\in\Sv_{1,1}$, $0\leq\gamma_n\leq1$, 
$$\gamma_n\rightharpoonup_\ast \gamma=\frac12\sum_{k\in\Z^d} \big|\phi(\cdot+k)\big\rangle\big\langle\phi(\cdot+k)\big|\quad\text{ in $L^\ii(\Omega,\B)$}$$
and
$$\rho_{\gamma_n}=\sum_{k\in\Z^d}\frac{1+\sin(2\pi n \omega_k)}2\;\big|\phi(\cdot+k)\big|^2\rightharpoonup \rho_\gamma=\frac12\sum_{k\in\Z^d} \big|\phi(\cdot+k)\big|^2$$
weakly in $L^p_s$ for $1< p<\ii$. We also have
$$\EE\bra{\int_Q\rho_{\gamma_n}}=\EE\bra{\int_Q\rho_{\gamma}},\qquad\forall n\in\NN$$
However, since $\sin(2\pi n \omega_k)\rightharpoonup0$ weakly but not strongly in $L^p([0,1])$, we do not have any strong convergence for $\rho_{\gamma_n}$.
\end{example}

\begin{proof}[Proof of Proposition~\ref{prop:weak-compactness}]
Consider a sequence $(\gamma_n)$ as in the statement. Since $(\gamma_n)$ is bounded in $L^\ii(\Omega,\B)$, there exists $\gamma\in L^\ii(\Omega,\B)$ such that $\gamma_n$ converges to $\gamma$ weakly-$\ast$ in $L^\ii(\Omega,\B)$, up to extraction of a subsequence (denoted the same for simplicity). Recall that $\gamma_n\wto_\ast\gamma$
means 
$$\lim_{n\to\ii}\EE\bra{\tr(A\gamma_n)}=\EE\bra{\tr(A\gamma)}$$
for all $A\in L^1(\Omega,\S_1)$. Using for instance $A=Y\,|f\rangle\langle g|$ for some fixed $f,g\in \H=L^2(\R^d)$ and some fixed $Y\in L^1(\Omega)$, we find in particular that 
\begin{equation}
\forall f,g\in\H,\ \forall Y\in L^1(\Omega),\quad \EE\bra{Y\pscal{g,\gamma f}}=\lim_{n\to\ii} \EE\bra{Y\pscal{g,\gamma_n f}}.
\label{eq:CV_weak_op} 
\end{equation}
Hence, $\pscal{g,\gamma_n f}$ converges to $\pscal{g,\gamma f}$ weakly$-\ast$ in $L^\ii(\Omega)$. 
Using this, it is easy to verify that $\gamma$ is ergodic and satisfies $\gamma^\ast=\gamma$, $0\leq\gamma\leq 1$ a.s.

Let now $(f_k)_{k\geq1}$ be any orthonormal basis of $L^2(Q)$ where we recall that $Q$ is the unit cell. Using that $\EE\bra{\pscal{f_k,\gamma_nf_k}}\to\EE\bra{\pscal{f_k,\gamma f_k}}$ for each $k\geq1$ as $n\to\ii$, and Fatou's lemma in $\ell^1(\NN)$, we obtain
\begin{eqnarray*}
 \EE\bra{\tr(\1_Q\gamma\1_Q)}=\sum_{k\geq1}\EE\bra{\pscal{f_k,\gamma f_k}}&\leq &\liminf_{n\to\ii} \EE\bra{\sum_{k\geq1}\pscal{f_k,\gamma_n f_k}}\\
&=&\liminf_{n\to\ii}\EE\bra{\tr(\1_Q\gamma_n\1_Q)}.
\end{eqnarray*}
By Proposition~\ref{magie}, we conclude that $\gamma\in\underline{\S}_1$.
The same argument can be employed to show that $\gamma\in\underline{\S}_{1,1}$,
assuming this time that each $f_k$ is in $H^1_0(Q)$. Then we have for each $k$
$$\lim_{n\to\ii}\EE\bra{\pscal{f_k,P_j\gamma_n P_jf_k}}=\lim_{n\to\ii}\EE\bra{\pscal{(P_jf_k),\gamma_n (P_jf_k)}}=\EE\bra{\pscal{(P_jf_k),\gamma (P_jf_k)}},$$
by~\eqref{eq:CV_weak_op} and with $P_j=-i\partial_{x_j}$. By Fatou's Lemma in $\ell^1(\NN)$ we see that
$$\tv{-\Delta\gamma}=\sum_{j=1}^d\sum_i \EE\bra{\pscal{(P_jf_k),\gamma (P_jf_k)}}\leq \liminf_{n\to\ii}\tv{-\Delta\gamma_n}.$$

Let us now prove that $\tv{\gamma_n}$ indeed converges to $\tv{\gamma}$. 
We consider a smooth function $\chi$ in $C^\ii_c(\R^d)$. The sequence $(\gamma_n)$ being bounded in $\Sv_{1,1}$, there exists  a constant $C\in\RR_+$ such that for all $n\in \NN$ and  $1\leq j\leq d$,
$$
 \EE\left( \tr\bra{\chi\gamma_n\chi} \right)+\EE\left( \tr(\chi P_j\gamma_n P_j\chi)\right)+\EE\left( \tr\bra{\partial_{x_j}\chi}\gamma_n\bra{\partial_{x_j}\chi} \right)\leq C.$$
Using again the relation $\left[ P_j,\chi\right]=-i\partial_{x_j}\chi$, we obtain
\begin{equation*}\label{eq_3}
\EE\left( \tr\left(  P_j\chi \gamma_n \chi P_j  \right)  \right)\leq 4C,
\end{equation*}
hence
\begin{align*}
\EE\left( \tr\left( \left(1-{\Delta} \right)^\frac{1}{2} \chi\gamma_n\chi \left(1-{\Delta} \right)^\frac{1}{2}  \right) \right) &= \EE\left(\tr\left( \chi\gamma_n \chi\right) \right) +\EE\left(\tr\left( -{\Delta}(\chi \gamma_n\chi)\right)\right)\\
&= \EE\left(\tr\left( \chi\gamma_n \chi\right) \right) +\sum_{j=1}^d\EE\left(\tr\left( P_j\chi \gamma_n \chi P_j \right)\right)\\
&\leq (1+4d) C.
\end{align*}
This proves that $\left(1-{\Delta} \right)^{1/2} \chi\gamma_n\chi \left(1-{\Delta} \right)^{1/2}$ is bounded in $L^1(\Omega,\S_1)$ or, equivalently, that $\left(1-{\Delta} \right)^{1/2} \chi\sqrt{\gamma_n}$ is bounded in $L^2(\Omega,\S_2)$.
From this we infer that
$$\left(1-{\Delta} \right)^\frac{1}{2} \chi\gamma_n\chi=\left\{\left(1-{\Delta} \right)^\frac{1}{2} \chi\gamma_n\chi \left(1-{\Delta} \right)^\frac{1}{2}\right\} \left(1-{\Delta} \right)^{-\frac{1}{2}}$$
is bounded in $L^1(\Omega,\S_1)$, since $\left(1-{\Delta} \right)^{-{1}/{2}}$ is a bounded operator.
Similarly, we can write
$$\left(1-{\Delta} \right)^\frac{1}{2} \chi\gamma_n\chi=\left\{\left(1-{\Delta} \right)^\frac{1}{2} \chi\sqrt{\gamma_n}\right\}\sqrt{\gamma_n}\chi,$$
which is now bounded in $L^2(\Omega,\S_2)$, since $\norm{\sqrt{\gamma_n}\chi}\leq C$ due to the assumption that $0\leq\gamma_n\leq1$. We conclude that $\left(1-{\Delta} \right)^{1/2} \chi\gamma_n\chi$ is bounded in $L^1(\Omega,\S_1)\cap L^2(\Omega,\S_2)$, hence in $L^p(\Omega,\S_p)$ for all $1\leq p\leq2$, by interpolation. In particular,
\begin{equation}
\left(1-{\Delta} \right)^{\frac12} \chi\gamma_n\chi\wto \left(1-{\Delta} \right)^{\frac12} \chi\gamma\chi\ \text{ weakly in $L^p(\Omega,\S_p)$ for all $1<p\leq2$.}
\label{eq:weak_CV_S_q} 
\end{equation}
That the limit can only be $\left(1-{\Delta} \right)^{1/2} \chi\gamma\chi$ follows for instance from~\eqref{eq:CV_weak_op} with functions $f,g\in H^1(\R^d)$.

We consider now a fixed function $Y\in L^{\ii}(\Omega)$ and write
$$\EE\bra{Y\tr(\chi\gamma_n\chi)}=\EE\bra{Y\tr\left(\left(1-{\Delta} \right)^{1/2} \chi\gamma_n\chi\1_B\left(1-{\Delta} \right)^{-1/2}\right)},$$
where $B$ is a large enough ball containing the support of $\chi$. By the Kato-Seiler-Simon inequality~\cite[Thm. 4.1]{trace_ideals},
\begin{equation}
\forall p\geq2,\qquad \norm{f(x)g(-i\nabla)}_{\S_p}\leq (2\pi)^{-d/p}\norm{f}_{L^p(\R^d)}\norm{g}_{L^p(\R^d)},
\label{eq:KSS}
\end{equation}
we have
$$\norm{\1_B\left(1-{\Delta} \right)^{-1/2}}_{\S_{1+d}}\leq (2\pi)^{-d/p}|B|^{\frac{1}{1+d}}\left(\int_{\R^d}\frac{dp}{(1+|p|^2)^{\frac{1+d}2}}\right)^{\frac{1}{1+d}},$$
hence $\1_B\left(1-{\Delta} \right)^{-1/2}\in\S_{1+d}$. Thus $Y\1_B\left(1-{\Delta} \right)^{-1/2}\in L^\ii(\Omega,\S_{1+d})\subset L^{1+d}(\Omega,\S_{1+d})$. Since $1<1+1/d\leq 2$ we obtain by the weak convergence~\eqref{eq:weak_CV_S_q} in $L^{1+1/d}(\Omega,\S_{1+1/d})=L^{1+d}(\Omega,\S_{1+d})'$, 
$$\lim_{n\to\ii}\EE\bra{Y\tr(\chi\gamma_n\chi)}=\EE\bra{Y\tr\left(\left(1-{\Delta} \right)^{1/2} \chi\gamma\chi\1_B\left(1-{\Delta} \right)^{-1/2}\right)}=\EE\bra{Y\tr(\chi\gamma\chi)}.$$
We can reformulate this into
\begin{equation}
\lim_{n\to\ii}\EE\bra{Y\int_\RRd\chi^2\rho_{\gamma_n}}=\EE\bra{Y\int_\RRd\chi^2\rho_{\gamma}},
\label{eq:prop-CV-rho} 
\end{equation}
for all $Y\in L^\ii(\Omega)$ and all $\chi\in C^\ii_c(\R^d)$.

As $\trv\left(-{\Delta} \gamma_n\right)$ is bounded, we infer from the Lieb-Thirring inequality for ergodic operators (Proposition~\ref{Lieb_T}) that $(\rho_{\gamma_n})$ is bounded in $L_s^{1+2/d}$. We can therefore extract a subsequence which weakly converges in $L_s^{1+2/d}$ to some $\rho \in L_s^{1+2/d}$. Since the space spanned by the functions of the form $Y\chi^2$ with $Y\in L^\ii(Q)$ and $\chi\in C^\ii_c(Q)$ is dense in $L^{1+d/2}(\Omega\times Q)$, we deduce from~\eqref{eq:prop-CV-rho} that $\rho_\gamma=\rho$. Now, using that $\1_Q\in L^{1+d/2}(\Omega\times Q)$, we finally obtain the claimed convergence 
\begin{equation}
\lim_{n\to\ii}\tv{\gamma_n}=\lim_{n\to\ii}\EE\bra{\int_Q\rho_{\gamma_n}}=\EE\bra{\int_Q\rho_{\gamma}}=\tv{\gamma}.
\end{equation}
This concludes the proof of the proposition.
\end{proof}

\subsubsection{Spectral projections of ergodic Schrödinger operators}

The following result provides a control of the average number of particles and kinetic energy per unit volume of the spectral projections of an ergodic Schr\"odinger operator, in terms of the negative component $V_-=\max(-V,0)$ of the external potential. We will use it later in Section~\ref{sec:SCF} to prove that the ground state density matrix of the reduced Hartree-Fock model with Yukawa potential is solution to a self-consistent equation.

\begin{proposition}[Spectral projections are in $\Sv_{1,1}$]\label{proj}
Let $V\in L^2_s$ be such that the operator $H=-{\Delta} +V$ is essentially self-adjoint on $C^\ii_c(\RRd)$ and $V_-\in L^{1+d/2}_s$. Denote by $P_{\lambda}=\1_{\left(-\ii,\lambda\right)}\bra{H}$ the spectral projection of $H$ corresponding to filling all the energy levels below $\lambda$. Then, $P_\lambda\in\Sv_{1,1}$ for any $\lambda\in\RR$ and there is a constant $C>0$ (depending only on $d\geq1$) such that 
\renewcommand{\labelenumi}{(\roman{enumi})}
\begin{equation}
\tv{ P_\lambda}\leq C\,\left(\EE\bra{\int_{Q}\bra{V-\lambda}_-^{\tfrac{d+2}{2}}\right)}^{\tfrac{d}{d+2}},
\label{eq:estim_trace_vol_lambda}
\end{equation}
and
\begin{equation}
\tv{-{\Delta} P_\lambda}\leq C\,\EE\bra{\int_{Q}\bra{V-\lambda}_-^{\tfrac{d+2}{2}}}. 
\label{eq:estim_kinetic_vol_lambda}
\end{equation}
\end{proposition}

The estimate~\eqref{eq:estim_trace_vol_lambda} on $\tv{ P_\lambda}$ is probably not optimal but it is sufficient for our purposes.

\begin{proof}
Let us first prove that $P_\lambda\in\Sv_{1,1}$ under the assumption that $V_-\in L^\ii\bra{\O\times \RRd}$. By the Feynman-Kac formula~\cite[Theorem 6.2 p.51]{FIbarry}, we have for all $t>0$
\begin{equation}\label{FK}
\rho_{e^{-t\left(-{\Delta}+V\right)}}\leq
\frac{e^{t\norm{V_-}_{L^\ii\bra{\RRd}}}}{\left(4\pi t\right)^{d/2}} \;\mbox{ a.s.}
\end{equation} 
Using then the inequality $\1_{\left(-\ii, \lambda\right]}\bra{x}\leq e^{-t(x-\lambda)}$ for all $\lambda\in\R$ and all $t>0$, as well as the fact that $V$ is uniformly bounded from below, we deduce that $P_\lambda\in\Sv_1$ and $\rho_{P_\lambda}\in L^\ii\bra{\O\times \RRd}$. Likewise, using the inequality  $x\1_{\left(-\ii, \lambda\right]}\bra{x}\leq \lambda\1_{\lambda\geq 0} e^{-t(x-\lambda)}$, we obtain that $HP_\lambda\in \Sv_1$, hence that $P_\lambda\in \Sv_{1,1}$. 

Now that we know that $P_\lambda\in\Sv_{1,1}$, we can derive bounds which only depend on $\norm{(V-\lambda)_-}_{L^{1+d/2}_s}$. The general case will then follow from a simple approximation argument. We start by noting that
\begin{eqnarray*}
0\leq\tv{\bra{-\Delta+V-\lambda}_-}&=&-\tv{\bra{-\Delta+V-\lambda}P_\lambda}\\
&\leq & -C\norm{\rho_{P_\lambda}}_{L_s^\frac{d+2}{d}}^\frac{d+2}{d}+ \norm{\rho_{P_\lambda}}_{L_s^\frac{d+2}{d}}\norm{\bra{V-\lambda}_-}_{L_s^\frac{d+2}{2}}\\
&\leq & C \norm{\bra{V-\lambda}_-}_{L_s^\frac{d+2}{2}}^\frac{d+2}{2},
\end{eqnarray*}
where we have used the Lieb-Thirring inequality~\eqref{eq:LT} for ergodic operators. Therefore 
 \begin{equation}\label{borne_trace}
\tv{ P_\lambda}=\norm{\rho_{P_\lambda}}_{L^1_s}\leq C \norm{\rho_{P_\lambda}}_{L_s^\frac{d+2}{d}}\leq C\norm{\bra{V-\lambda}_-}_{L_s^\frac{d+2}{2}}^\frac{d}{2}.
\end{equation} 
As $0\leq \tv{\bra{-\Delta+V-\lambda}_-}$, we obtain
 \begin{equation}\label{borne_Ec}
\tv{-\Delta P_\lambda}\leq C \norm{\bra{V-\lambda}_-}_{L_s^\frac{d+2}{2}}^\frac{d+2}{2}.
\end{equation} 
This concludes the proof in the case of bounded below potentials. In the general case we consider the sequences of cutoff potentials $V_n=\max\set{V, -n}$ and corresponding operators $H_n=-\Delta +V_n$ and show that for any bounded continuous function $f$, the operator $f\bra{H_n}$ converges to $f\bra{H}$ in the strong operator topology a.s. We conclude the proof using an appropriate approximation of $\1_{ \left(-\ii,\lambda\right] }$ by bounded continuous functions (see~\cite{these} for details).
\end{proof}
\renewcommand{\labelenumi}{\arabic{enumi}.}

We can now use the previous theorem to deduce a useful variational characterization of the spectral projection $P_\lambda$ among all ergodic fermionic density matrices $\gamma\in\K$ having a locally finite kinetic energy.

\begin{proposition}[Variational characterization of spectral projections]\label{prop:var_charac}
Assume that $V$ is as in Proposition~\ref{proj} and denote again $P_\lambda:=\1_{(-\ii,\lambda)}(H)$ with $H=-\Delta+V$.
For every $\lambda\in\R$, the minimization problem
\begin{equation}
\inf_{\gamma\in\K}\left\{\trv (-\Delta\gamma)+\EE\bra{\int_QV\,\rho_\gamma}-\lambda\,\tv{\gamma}\right\}
\label{eq:def_var_spec_proj}
\end{equation}
admits as unique minimizers the operators of the form
$\gamma=P_\lambda+\delta$
where $0\leq\delta\leq \1_{\{\lambda\}}(H)$.
\end{proposition}

Note that $\EE(\int_QV\,\rho_\gamma)$ is well defined in $(-\ii,+\ii]$ since $V_-\in L^{1+d/2}_s$ by assumption, whereas $\rho_\gamma\in L^{1+2/d}_s$ by the Lieb-Thirring inequality~\eqref{eq:LT}.

\begin{proof}
When $\gamma$ is smooth enough ($-\Delta\gamma\in\Sv_1$ for example) and $V\in L^\ii_s$, we can write
\begin{multline*}
\trv \bra{-\Delta(\gamma-P_\lambda)}+\EE\bra{\int_QV\,\left(\rho_\gamma-\rho_{P_\lambda}\right)}-\lambda\tv{\gamma-P_\lambda}\\
=\trv \bra{(-\Delta+V-\lambda)(\gamma-P_\lambda)}\geq\trv \bra{|-\Delta+V-\lambda|(\gamma-P_\lambda)^2}. 
\end{multline*}
In the last estimate we have used the cyclicity property~\eqref{cyclicity} and the fact that
$$P_\lambda^\perp(\gamma-P_\lambda)P_\lambda^\perp-P_\lambda(\gamma-P_\lambda)P_\lambda\geq (\gamma-P_\lambda)^2,$$
which turns out to be equivalent to $0\leq\gamma\leq1$. A simple approximation argument now shows that the inequality
\begin{multline*}
\trv (-\Delta\gamma)+\EE\bra{\int_QV\,\rho_\gamma}-\lambda\tv{\gamma}
\geq \trv (-\Delta P_\lambda)+\EE\bra{\int_QV\,\rho_{\P_\lambda}}-\lambda\tv{P_\lambda}\\+\tv {|H-\lambda|^{1/2}(\gamma-P_\lambda)^2|H-\lambda|^{1/2}}
\end{multline*}
is actually valid under the weaker assumptions of the proposition. It is then clear that $P_\lambda$ minimizes~\eqref{eq:def_var_spec_proj} and that the other minimizers must satisfy $|H-\lambda|^{1/2}(\gamma-P_\lambda)=0$, which is the same as saying that the range of $\gamma-P_\lambda$ is included in the kernel of $H-\lambda$.
\end{proof}

\subsubsection{A representability criterion}\label{repre}
The aim of representability criteria is to identify sets of densities $\rho$ that arise from admissible density matrices. For finite systems, if $\gamma\in \S_1 \cap {\mathcal S}$,  $0\leq \gamma\leq 1$, and $\tr\bra{-\Delta \gamma}<\ii$, then $\rho_\gamma\geq 0$ and $\sqrt{\rho_\gamma}\in H^1\bra{\RRd}$ by the Hoffmann-Ostenhof inequality. Lieb's representability theorem~\cite[Theorem 1.2]{Lieb} shows that these conditions are sufficient for a function $\rho$ to be representable. 

In the ergodic case, we know that a density $\rho$ must satisfy $\rho\geq0$, $\sqrt{\rho}\in H^1_s$ and $\rho\in L^{1+2/d}_s$, by the Lieb-Thirring inequality~\eqref{eq:LT}. Clearly a stationary function $\rho$ such that $\rho\geq0$ and $\sqrt{\rho}\in H^1_s$ is not necessarily the density of an ergodic density matrix with finite kinetic energy, since in general
$$\set{\rho \ge 0\;|\;\sqrt{\rho}\in H^1_s}\not\subset L^\frac{d+2}{d}_s.$$ 
It is an interesting open problem to determine the exact representability conditions in the ergodic case. Theorem~\ref{n-repre} below gives sufficient conditions for $\rho$ to be representable. These conditions are also necessary for $d=1$.

\begin{theorem}[A sufficient condition for representability]\label{n-repre}
We assume that $d\geq 1$.
Let $\rho$ be a function satisfying 
$$\rho\geq 0,\;  \rho\in L^3_s\;\mbox{and}\; \sqrt{\rho}\in H_s^1.$$ 
Then, there exists a self-adjoint operator $\gamma$ in $\Sv_{1,1} \cap \underline{\mathcal S}$, satisfying $0\leq \gamma \leq 1$ and $\rho_\gamma=\rho$ a.s. 
\end{theorem}

Theorem~\ref{n-repre} is proved in the Appendix, following ideas of Lieb~\cite{Lieb}.


\section{Yukawa and Coulomb interaction}\label{sec:Coulomb-Yukawa}

This section is devoted to the definition of the potential energy per unit volume of a stationary charge distribution $f$. In our setting, $f$ will be $\rho_\gamma-\mu$, where $\mu$ is the nuclear charge distribution and $\rho_\gamma$ the density associated with an electronic state $\gamma$. We will consider two types of interactions, namely the (long-range) Coulomb and the (short-range) Yukawa interactions.

In dimension $d\geq1$, the Coulomb self-interaction of a charge density $f$ is given by
$$D(f,f)=\int_{\R^d}V(x)\,f(x)\,dx.$$
where $V$ is the Coulomb potential induced by $f$ itself, which is solution to Poisson's equation
\begin{equation}
-\Delta V=|S^{d-1}|\;f.
\label{eq:Poisson} 
\end{equation}
Here $|S^{d-1}|$ is the Lebesgue measure of the unit sphere $S^{d-1}$ ($|S^0|=2$, $|S^1|=2\pi$, $|S^2|=4\pi$).
For later purposes, it is convenient to regularize this equation by adding a small mass $m$ as follows :
\begin{equation}
(-\Delta+m^2) V=|S^{d-1}|\;f.
\label{eq:Poisson-Yukawa}
\end{equation}
Whenever $m=0$ or $m>0$, we have the following formulas for the Coulomb ($m=0$) and Yukawa ($m>0$) self-energies:
\begin{align}
D_m(f,f)=\av{S^{d-1}}\int_\RRd\frac{\av{\widehat{f}\bra{K}}^2}{\av{K}^2+m^2}\,dk&=\av{S^{d-1}}\norm{\bra{-\Delta+m^2}^{-\frac{1}{2}}f}_{L^2}^2\label{eq:Dm_A}\\
&=\int_{\R^d}\int_{\R^d}Y_m(x-y)f(x)\,f(y)\,dx\,dy\label{eq:Dm_B}\\
&=\int_{\R^d}\left|\int_{\R^d}W_m(x-y)f(y)\,dy\right|^2dx.\label{eq:Dm_C}
\end{align}
Here $\widehat{f}$ is the Fourier transform\footnote{In the whole paper we use the convention $\widehat{f}\bra{K}=\bra{2\pi}^{-\frac{d}{2}}\int_\RRd f\bra{x}e^{-iK\cdot x}dx$.} of $f$. 
Of course we need appropriate decay and integrability assumptions on $f$ to make the previous formulas meaningful. 
The Yukawa and Coulomb kernels are given by
$$
Y_m(x)=\left\{
\begin{array}{ll}
m^{-1}e^{-m\left|x \right|}\\
K_0\bra{m\av{x}}\\
|x|^{-1}e^{-m\left|x \right|}\\
\end{array}
\right.
\quad\text{and}\quad
Y_0(x)=\left\{
\begin{array}{ll}
-\av{x}\; & \mbox{if}\;d=1,\\ -\log\bra{\av{x}}  \;  &\mbox{if}\;d=2,\\ |x|^{-1} \; &\mbox{if}\;d=3,
\end{array}
\right.
$$
with $K_0\bra{r}=\int_0^\ii e^{-r\cosh t}\,dt$ the modified Bessel function of the second type~\cite{LL}.
The Coulomb potential is nothing but the limit of the Yukawa potential when the parameter $m$ goes to $0$.
Similarly, the function $W_m$ is defined by its Fourier transform
$$\widehat{W_m}(K)=\frac{\sqrt{|S^{d-1}|}(2\pi)^{-d/2}}{\sqrt{m^2+|K|^2}}.$$
Using the integral representation $x^{-1/2}=2\pi^{-1}\int_0^\ii (x+s^2)^{-1}ds$, we see that 
\begin{equation}\label{eq:Wm}
W_m(x)=\frac{2}{\pi\sqrt{|S^{d-1}|}}\int_0^\ii Y_{\sqrt{s^2+m^2}}(x)\,ds.
\end{equation}
This can be used to compute $W_m$ in some cases, or to simply deduce that, when $m>0$, $W_m$ is positive, decays exponentially at infinity, and behaves at zero like $|x|^{-2}$ in dimension $d=3$, like $|x|^{-1}$ when $d=2$ and like $\log|x|$ for $d=1$.

Our goal in this section is to define the Yukawa and Coulomb energies per unit volume for a stationary charge distribution $f$. Formally, this is just 
$$\EE\bra{\int_Q V\,f}$$
where $V$ solves~\eqref{eq:Poisson} for $m=0$ or~\eqref{eq:Poisson-Yukawa} for $m>0$. We are implicitly using here the fact that the potential $V$ is stationary when $f$ has this property. Unfortunately, giving a meaning to Poisson's equation~\eqref{eq:Poisson} in the stochastic setting in not an easy task. Already when $f$ is periodic, we know that this equation can only have a solution when $\int_Qf=0$. Here the situation is even worse, as we explain below. To simplify matters, we first introduce the Yukawa energy per unit volume $D_m(f,f)$ for $m>0$ and then we define the Coulomb energy as the limit of $D_m(f,f)$ as $m\to0$, when it exists. Thus we start by giving a clear meaning to the three possible formulas~\eqref{eq:Dm_A},~\eqref{eq:Dm_B} and~\eqref{eq:Dm_C} in the Yukawa case $m>0$. In the next section we introduce the \emph{stationary Laplacian} $-\Delta_s$ which allows to write a formula similar to~\eqref{eq:Dm_A}.

\subsection{The stationary Laplacian}

In this section we define an operator which we call the \emph{stationary Laplacian}, which is nothing but the usual Laplacian in the $x$ variable acting on $L^2(\Omega\times Q)$, with stationary boundary conditions at the boundary of $Q$. Surprisingly, this operator does not seem to have been considered before.

Let $A_0$ be the operator on $L^2_s$ defined by
$$
\left\{ 
\begin{array}{l}
D(A_0)=L^2_s \cap L^2(\Omega, C^2(\RRd)), \\
A_0f(\omega,x)= -\Delta f(\omega,x)\;  \mbox{a.s. and a.e.},\;\forall f\in D(A_0) ,
\end{array}
\right.
$$
where $\Delta=\sum_{j=1}^d\partial_{x_j}^2$ refers to the usual Laplace operator on $C^2(\mathbb R^d)$ w.r.t. the $x$ variable. Using stationarity, we obtain 
$$
\langle f,A_0g\rangle_{L^2_s} = \mathbb E \left(\int_Q \overline{\nabla  f(\cdot,x)}\cdot\nabla g(\cdot,x)\,dx\right),\;\forall f,g \in D(A_0) .
$$
Thus,  $A_0$ is a symmetric, non-negative operator on $L^2_s$ with dense domain $D(A_0)$. We denote by $-{\Delta}_s$ its Friedrichs extension~\cite[theorem 4.1.5, p.115]{caught}, and call the operator $\Delta_s$  the stationary Laplacian. The form domain of the operator $-\Delta_s$ is $H^1_s$ and its domain is $H^2_s$. 

When $\O$ is finite, the spectrum of $-\Delta_s$ is purely discrete. If the probability space is defined as in Example~\ref{example1}, then $\sigma(-\Delta_s)=[0,+\infty)$. Thanks to the ergodicity of the group action, one can prove that $\mbox{ker}\bra{-\Delta_s}=\mbox{span}\set{1}$. In contrast to the periodic case, there is (in general) no gap in the spectrum of $-\Delta_s$ above $0$. In other words, there is no Poincaré-Wirtinger type inequality in $H^1_s$. This can be seen,  for instance, by considering the sequence of functions $\Phi_n(\o,x)=n^{-d/2}\Phi\bra{\o,n^{-1}x}$, where $\Phi\bra{\o,x}=\sum_{k\in\ZZd}Y(\tau_k(\o))\chi(x-k)$ with $Y\in L^1\bra{\O}$ and $\chi\in C^\ii_c\bra{\RRd}$ with support in the unit cube $Q$ and such that $\int_Q\chi(x)\,dx=0$. These functions are such that $\norm{\Phi_n}_{L^2_s}=1$ and $\EE(\int_Q \Phi_n)=0$ for any $n\in\NN$, and $\norm{\nabla\Phi_n}_{\bra{L^2_s}^d} \rightarrow 0$ as $\cn$.  

That there is no Poincaré-Wirtinger inequality means that solving Poisson's equation~\eqref{eq:Poisson} in the stochastic (ergodic) setting is complicated. Contrarily to the periodic case, it is \emph{not} sufficient to ask that $f\notin\mbox{ker}\bra{-\Delta_s}$, that is $\EE(\int_Qf)=0$. If we are given $f\in L^2_s$, then we see that there exists $V\in L^2_s$ such that $-\Delta_s V=|S^{d-1}|f$ if and only if $f$ belongs to the range of $-\Delta_s$. In the next section we consider the simpler Yukawa equation~\eqref{eq:Poisson-Yukawa}.

\subsection{The Yukawa interaction}

Let $m>0$. If $f\in L^2_s$, we can define by analogy with~\eqref{eq:Dm_A}
\begin{equation}
D_m(f,f)=\av{S^{d-1}}\norm{\bra{-\Delta_s+m^2}^{-\frac{1}{2}}f}_{L^2_s}^2.
\label{eq:first-def-Dm} 
\end{equation}
The operator $\bra{-\Delta_s+m^2}^{-\frac{1}{2}}$ being bounded, $D_m$ is well-defined on $L^2_s$. To set up our mean-field model for disordered crystals, we however need to extend the quadratic form $D_m$ to a larger class of functions. Formal manipulations show that for a stationary function $f$
\begin{align}
D_m(f,f)&=\EE\bra{\int_Q\int_{\R^d}Y_m(x-y)\,f(\cdot,x)\,f(\cdot,y)\,dx\,dy}\nonumber\\
&=\EE\bra{\int_{Q}\left|\int_{\R^d}W_m(x-y)f(\cdot,y)\,dy\right|^2dx}.\label{eq:formal_Dm}
\end{align}
The second formula is more suitable for a proper definition of $D_m$.
We claim that the function $\bra{W_m*f}\bra{\o,x}:=\int_{\RRd}f\bra{\o,y}W_m\bra{x-y}\,dy$ is well-defined for all $f \in L^1_s$, and is in $L^1_s$. This follows from the following elementary result.

\begin{lemma}[Convolution of stationary functions]\label{lem:convolution}
Let $f\in L^t_s(L^q)$ and $W\in L_{\rm loc}^p(\R^d)$ such that
$$\sum_{k\in\Z^d}\norm{W}_{L^p(Q+k)}<\ii,$$
for some $1\leq p,q,t\leq\ii$. Then the function
\begin{equation}
\bra{W*f}\bra{\o,x}:=\int_{\RRd}f\bra{\o,y}W\bra{x-y}\,dy
\label{eq:def_convolution} 
\end{equation}
belongs to $L^t_s(L^r)$ with $1+1/r=1/p+1/q$, and
\begin{equation}\label{eq:convolution}
\norm{W*f}_{L^q_s(L^r)}\leq C\norm{f}_{L^t_s(L^q)}\left(\sum_{k\in\Z^d}\norm{W}_{L^p(Q+k)}\right)
\end{equation}
for a  constant $C$ depending only on the dimension $d$.
If $1<p,q,r<\ii$, we can replace $\norm{W}_{L^p(Q+k)}$ by the weak norm $\norm{W\1_{Q+k}}_{L^p_w}$ in~\eqref{eq:convolution}.
\end{lemma}

\begin{proof}
In order to prove the convergence of the integral in~\eqref{eq:def_convolution}, we write 
\begin{align*}
\int_{\RRd}|f\bra{\o,y}|\,|W\bra{x-y}|\,dy&=\sum_{k\in\Z^d}\int_{Q+k}|f\bra{\o,y}|\,|W\bra{x-y}|\,dy\\
 &=\sum_{k\in\Z^d}\int_{Q}|f\bra{\tau_k(\o),y}|\,|W\bra{x-y-k}|\,dy
\end{align*}
where we have used the stationarity of $f$. By the standard Young inequality we have for a.e. $x\in Q$ and a.s.
$$\norm{\int_{Q}|f\bra{\tau_k(\o),y}|\,|W\bra{x-y-k}|\,dy}_{L^r(Q)}\leq \norm{W(\cdot-k)}_{L^p(2Q)}\norm{f(\tau_k(\o),\cdot)}_{L^q(Q)}$$
and therefore
$$\norm{\int_{Q}|f\bra{\tau_k(\o),y}|\,|W\bra{x-y-k}|\,dy}_{L^t_s(L^r)}\leq \norm{W(\cdot-k)}_{L^p(2Q)}\norm{f}_{L^t_s(L^q)}.$$
The rest follows. The estimate with the weak norm $\norm{W\1_{Q+k}}_{L^p_w}$ follows from the generalized Young inequality~\cite{ReeSim2}.
\end{proof}

Since $W_m\in L^1(\R^d)$ when $m>0$, Lemma~\ref{lem:convolution} shows that $W_m*f\in L^1_s$ when $f\in L^1_s$. Now we can define
\begin{equation}
\boxed{D_m(f,f):=\EE\bra{ \int_{Q} |W_m*f|^2}} 
\label{eq:def_Dm}
\end{equation}
for any $f$ in the space 
\begin{equation}
\boxed{\D_Y:=\left\{f\in L^1_s\;|\; W_m*f\in L^2_s\right\}}
\label{eq:def_DY} 
\end{equation}
which we call the space of locally integrable functions with \emph{locally finite Yukawa energy}. It is easy to see that the space $\D_{Y}$ in fact does \emph{not} depend on $m>0$. It is a subspace of $L^1_s$ , with associated norm $\norm{f}_{L^1_s}+D_1(f,f)^{1/2}$, and a Banach space for this norm.

Using Lemma~\ref{lem:convolution},~\eqref{eq:Wm} and the known properties of $W_m$, we deduce the following result.

\begin{corollary}[Some functions of $\D_Y$]\label{cor:in_D_Y}
We have, in dimension $d$,
$$L^1_s\supset \D_Y\supset \begin{cases}
L^2_s(L^1) & \text{if $d=1$,}\\
L^2_s(L^q),\ \forall q>1 & \text{if $d=2$,}\\
L^2_s(L^{6/5}) & \text{if $d=3$.}
\end{cases}$$
\end{corollary}

When $f\in L^2_s$ we have that both $Y_m*f$ and $W_m*f$ belong to $L^2_s$, since $Y_m$ and $W_m$ are in $L^1(\R^d)$ when $m>0$. Thus we always have $f\in \D_Y$ and it is then an exercise to show that all the formulas for $D_m(f,f)$ in~\eqref{eq:first-def-Dm} and~\eqref{eq:formal_Dm} make sense and coincide. Indeed, we have 
$$Y_m*f=|S^{d-1}|\bra{-\Delta_s+m^2}^{-1} f\quad\text{and}\quad W_m*f=\sqrt{|S^{d-1}|}\,\bra{-\Delta_s+m^2}^{-1/2} f.$$

\subsection{The Coulomb interaction}

As mentioned previously, the Coulomb potential can be seen as the limit of the Yukawa potential when the parameter $m>0$ goes to zero. More precisely, as $0\leq (-{\Delta}_s+m^2)^{-1/2}\leq (-{\Delta}_s+m'^2)^{-1/2} $ for all $m \geq m' > 0$, the function $m\mapsto D_m\bra{f,f}$ is  non-increasing on $(0,+\infty)$, for any $f\in \D_Y$. It would therefore be natural to define the average Coulomb energy per unit volume as the limit of $D_m\bra{f,f}$ when $m\to0$, but we will proceed slightly differently. 

To simplify some later arguments, we define the Coulomb energy per unit volume by compensating the charge by a jellium background. This means we introduce for a stationary charge distribution $f \in \D_Y$
$$
\boxed{D_0(f,f) := \lim_{m\rightarrow 0}D_m\left(f-\EE\bra{\int_Qf}\,,\,f-\EE\bra{\int_Qf}\right),}
$$
together with the associated space
$$
\boxed{\D_C=\set{f\in \D_Y\;|\;\lim_{m\rightarrow 0}D_m\left(f-\EE\bra{\int_Qf}\,,\,f-\EE\bra{\int_Qf}\right)<\ii},}
$$
of the locally integrable stationary charge distributions $f$ with locally finite Coulomb energy (when compensated by a jellium background). We again emphasize that $\D_C\subset L^1_s$ by construction.

When $f\in L^2_s\cap\D_Y$, the limit is finite if and only if $f-\EE(\int_Qf)$ belongs to the quadratic form domain of $(-\Delta_s)^{-1}$, and we have
by the functional calculus
\begin{align*}
D_0(f,f)&=\lim_{m\rightarrow 0}D_m\left(f-\EE\bra{\int_Qf}\,,\,f-\EE\bra{\int_Qf}\right)\\
&= \av{S^{d-1}}\norm{\bra{-\Delta_s}^{-\frac{1}{2}}\left(f-\EE\bra{\int_Qf}\right)}_{L^2_s}^2. 
\end{align*}
For $f$ only in $\D_C$, the family $(-\Delta_s+m^2)^{-1/2}(f-\EE(\int_Qf))$ is Cauchy in $L^2_s$ when $m$ goes to zero and we still denote its limit by $(-\Delta_s)^{-1/2}(f-\EE(\int_Qf))$.

The following result means, in particular, that in the physically relevant case $d = 3$, a stationary function $f \in L^2_s(L^{6/5})$ whose charge and dipole moment in the unit cell $Q$ vanish a.s., has a finite average Coulomb energy per unit volume.

\begin{proposition}[Some functions in $\D_C$]\label{example3}  Let $d \le 3$ and $f$ be a function of 
$L^2_s\bra{L^q}$, with $q=1$ if $d=1$, $q>1$ if $d=2$ and $q=\frac{6}{5}$ if $d=3$, such that
\begin{equation}\label{eq:1m0}
q\bra{\o}=\int_Qf\bra{\o,x}\,dx=0\quad\mbox{and}\quad p\bra{\o}=\int_Qxf\bra{\o,x}\,dx=0\quad \mbox{a.s.}
\end{equation}
Then, $f \in \D_C$. 
\end{proposition}

\begin{proof}
For the sake of brevity, we only detail the proof for $d=3$. Let $f$ be a function of $L^2_s\bra{L^q}$ satisfying~\eqref{eq:1m0}. As $\EE(\int_Qf)=0$, we have for all $m > 0$,
\begin{align*}
 D_m\left(f-\EE\bra{\int_Qf}\,,\,f-\EE\bra{\int_Qf}\right) &  = \dps\EE \left(\int_Q\int_\RRd \frac{e^{-m\av{x-y}}}{\av{x-y}}f(\cdot, y)f(\cdot,x)\,dy\,dx \right) \\
&\dps \le C \norm{f}_{L^2_s(L^\frac{6}{5})}^2 + \EE \left( \sum_{k\in\ZZ^3, \, |k| \ge 3} A_{m,k} \right), 
\end{align*}
where
$$
 A_{m,k}(\o) = \int_{Q\times Q}\frac{e^{-m\av{k+y-x}}}{\av{k+y-x}} f\bra{\o,y}f\bra{\tau_k\bra{\o},x} \,dy \,dx.
$$
Noticing that for all $m > 0$, $(x,y) \in Q \times Q$, and $k \in \ZZ^d$ such that $|k| \ge 3$,   
$$
\left| e^{-m|k+y-x|}-e^{-m|k|} + m e^{-m|k|} (|k+y-x|-|k|) \right| \le m^2 e^{-m|k|/2},
$$
and using the fact that $q(\omega)=0$ a.s. we obtain that
$$
A_{m,k}(\o) = (1+m|k|)e^{-m|k|} B_{m,k}(\o) + C_{m,k}(\o) \; \mbox{ a.s.},
$$
with
$$
B_{m,k}(\o) = \int_{Q\times Q}\frac{f\bra{\o,y}f\bra{\tau_k\bra{\o},x}}{\av{k+y-x}}\,dy \,dx
$$
and
$$
|C_{m,k}(\o)| \le 2 m^2 |k|^{-1} e^{-m|k|/2} \|f(\o,\cdot)\|_{L^1(Q)}  \|f(\tau_k(\o),\cdot)\|_{L^1(Q)}.
$$
It then follows from (\ref{eq:1m0}) that
\begin{equation*}
B_{m,k}(\o)  = \int_{Q\times Q}f(\o,x)f(\tau_k\bra{\o},y)F(k, x-y)\,dx\,dy \; \mbox{ a.s.},
\end{equation*}
where
$
F(k,h)=\frac{1}{\left| k+h \right|} 
- \left(  
\frac{1}{\left| k \right|}
-\frac{e_k\cdot h} {\left| k \right|^2 } 
+ \frac{3\left(e_k\cdot h\right)^2-\left|h\right|^2 }{2\left| k \right|^3}    \right)$ and  $e_k=k/|k|$. Thanks to the multipole expansion formula (see e.g. \cite[Lemma 9]{mathieu}), there exists a constant $C$ such that for all $k\in \RR^3\setminus\{0\}$ and $h\in \RR^3$ with $k+h\neq 0$,
$
\left| F(k,h) \right| \leq \frac{C\left|h\right|^3}{\left|k\right|^3\left|k+h\right|}.
$
Therefore,
$$
|B_{m,k}(\o)| \le  C |k|^{-4} \left\| f\bra{\o,\cdot} \right\|_{L^1(Q)} \left\| f\bra{\tau_k\bra{\o},\cdot}\right\|_{L^1(Q)}  \; \mbox{ a.s.}
$$
Consequently,
$$
\EE \left( \sum_{k\in\ZZ^3, \, |k| \ge 3} A_{m,k} \right)
\le C \|f\|_{L^2_s(L^1)}^2 \, \left( \sum_{k\in\ZZ^3 \setminus \left\{0\right\}}
 \frac{1}{|k|^{4}} + \sum_{k\in\ZZ^3 \setminus \left\{0\right\}} m^2  \frac{e^{-m|k|/2}}{|k|} \right),
$$
from which we infer that 
$$
D_m\left(f-\EE\bra{\int_Qf}\,,\,f-\EE\bra{\int_Qf}\right) \le C \|f\|_{L^2_s(L^{\frac 65})}^2 \left( 1 + \sum_{k\in\ZZ^3 \setminus \left\{0\right\}} m^2  \frac{e^{-m|k|/2}}{|k|} \right),
$$
for a constant $C$ independent of $f$. As 
$$
\lim_{m \to 0^+} \sum_{k\in\ZZ^3 \setminus \left\{0\right\}} m^2  \frac{e^{-m|k|/2}}{|k|} = \int_{\R^3} \frac{e^{-|x|/2}}{|x|} \, dx < \infty,
$$
we finally obtain that $f\in\D_C$.
\end{proof}

The proof of Proposition~\ref{example3} can be adapted to show that $D_0(f,f)$ is the limit of the supercell Coulomb energy per unit volume (see Section~\ref{limit_thermo} and~\cite{CDL}), for any fixed $f$ satisfying the neutrality assumptions~\eqref{eq:1m0}. It is an open problem to prove the same for the functions $f\in \D_C$ which do not satisfy~\eqref{eq:1m0}.

\subsection{Dual characterization}

The purpose of this section is to provide a useful characterization of the Yukawa and Coulomb spaces $\D_Y$ and $\D_C$ by duality. Let us introduce 
the spaces of test functions
$$E_Y=\mbox{span}\left\{ \Phi_{\chi,Y} \mbox{ with } Y\in L^\infty(\O),\;  \chi\in \mathcal S\bra{\RRd} \right\},$$ 
and
$$E_C=\mbox{span}\left\{ \Phi_{\chi,Y} \mbox{ with } Y\in L^\infty(\O),\;  \chi\in \mathcal S_0\bra{\RRd} \right\},$$ 
where $\mathcal S\left(\RRd\right)$ is the Schwartz space, $\mathcal S_0\bra{\RRd}=\left\{ \chi\in \mathcal{S}(\RRd)\;|\;  \widehat{\chi}\in C^\infty_c(\RRd\setminus\{0\})  \right\}$, 
and 
$$\Phi_{\chi,Y}(\o,x)=\sum_{k\in\ZZd}Y(\tau_k(\o))\chi(x-k),\;\mbox{a.s. and a.e.}$$

The following says that $E_Y$ (resp. $E_C$) are dense in $L^p_s$ (resp. in $L^p_s\cap\text{ker}(-\Delta_s)^\perp$).

\begin{lemma}[Density of $E_Y$ and $E_C$]\label{density_result}
For any $1 \le p < \ii$, the set  $E_Y$ is dense in $L^p_s$ and the set $E_C$ is dense in $\widetilde L^p_s=\set{f\in L^p_s\;|\;\EE\bra{\int_Qf}=0}$.
\end{lemma}

\begin{proof}
We sketch here the proof of the density of $E_C$ in $\widetilde L^p_s$, 
and refer the reader to~\cite{these} for further details. Let 
$\phi \in (\widetilde L^p_s)'= \widetilde L^{p'}_s$, where $p'=(1-p^{-1})^{-1}$ is the conjugate exponent of $p$, be such that 
\begin{equation}\label{ortho}
\forall Y \in L^\infty(\O), \; \forall \chi \in {\mathcal S}_0(\R^d), \quad
\EE\bra{\int_Q \phi\, \Phi_{\chi,Y}} = 0.
\end{equation}
For $Y \in L^\infty(\O)$, we denote $f_Y(x)=\EE\left( Y(\cdot) \phi(\cdot,x) \right)$. The function $f_Y$ is in $L^{p'}_{\rm unif}(\RRd)$, hence it is a tempered distribution: $f_Y\in\mathcal{S}'(\R^d)$. In view of~\eqref{ortho}, we have for all $\chi\in {\mathcal S}_0$, $\langle \F^{-1}(f_Y), \widehat{\chi}\rangle_{\mathcal S'\bra{\RRd},\mathcal S\bra{\RRd}}=0$. Therefore $\F^{-1}(f_Y)$ is supported in $\{ 0\}$, which implies that
$$
\F^{-1}(f_Y)=\sum_{|\alpha|\leq N}c_\alpha \partial^\alpha\delta_0,
$$ 
with $N\in \NN$ and $c_\alpha \in \CC$. It follows that 
$$
f_Y(x)=\sum_{|\alpha|\leq N}\widetilde{c}_\alpha x^\alpha,
$$
with $\widetilde{c}_\alpha \in \CC$. As $f_Y$ is in $L^{p'}_{\rm unif}(\RRd)$, all the coefficients $\widetilde c_\alpha$ are equal to zero, except possibly $\widetilde c_0$, and $f_Y$ is a constant. As $Y\mapsto f_Y$ is a continuous linear form on $L^p(\O)$, there exists $Z\in L^{p'}(\O)$ such that for all $Y\in L^\infty(\O)$:
$\EE\left(YZ\right)=f_Y$. 
It follows that for all $x\in\RRd$, $\phi(\o,x)=Z(\o)$. We know that any stationary function independent of $x$ is a.s. and a.e. constant~\cite{pastur}. As $\EE(\int_Q \phi(\o,x))=0 $, we conclude that $\phi=0$, which proves that $E_C$ is dense in $\widetilde L^p_s$.
\end{proof}

It can be verified~\cite{these} that
\begin{equation}
\forall \Phi_{\chi,Y}\in E_Y,\quad \bra{-\Delta_s+1}^{-\frac{1}{2}}\Phi_{\chi,Y}=\Phi_{(-\Delta+1)^{-1/2}\chi,Y}
\label{eq:action_Laplacian_Y} 
\end{equation}
and, similarly, that
\begin{equation}
\forall \Phi_{\chi,Y}\in E_C,\quad \bra{-\Delta_s}^{-\frac{1}{2}}\Phi_{\chi,Y}=\Phi_{(-\Delta)^{-1/2}\chi,Y}. 
\label{eq:action_Laplacian_C}
\end{equation}
A straightforward consequence of Lemma~\ref{density_result} and~\eqref{eq:action_Laplacian_Y}-\eqref{eq:action_Laplacian_C} is the following

\begin{corollary}[Dual characterization of $\D_Y$ and $\D_C$]\label{duality} Let $f \in L^1_s$.
\renewcommand{\labelenumi}{(\roman{enumi})}	
\begin{enumerate}
\item If $(-\Delta_s+1)^{-{1}/{2}}f$, seen as a linear form on $E_Y$, is continuous on \\ $(E_Y,\|\cdot\|)_{L^2_s}$, then $f\in\D_Y$ and $D_m(f,f)=|S^{d-1}|\|(-\Delta_s+m^2)^{-{1}/{2}}f\|^2_{E_Y^*}.$
\item If $\EE(\int_Q f)=0$ and $(-\Delta_s)^{-{1}/{2}}f$, seen as a linear form on $E_C$, is continuous on $(E_C,\|\cdot\|_{L^2_s})$, then $f\in\D_C$ and $D_0(f,f)=|S^{d-1}|\|(-\Delta_s)^{-{1}/{2}}f\|^2_{E_C^*}.$

\end{enumerate}
\end{corollary}

\section{Stationary reduced Hartree-Fock model}\label{sec:rHF}

After these long preliminaries, we now introduce and study a reduced Hartree-Fock (rHF) model for crystals with nuclear charges randomly distributed following a stationary function $\mu\geq 0$. We typically think of $\mu$ being of the form
$$\mu(\omega,x)=\sum_{k\in\Z^d}q_k(\omega)\chi(x-k-\eta_k(\omega))$$
with $\int\chi=1$ and which describes a lattice of nuclei whose charges and positions are perturbed in an i.i.d. ergodic fashion. However in this work we do not want to restrict ourselves to $\mu$'s of this very specific form and for us $\mu$ is any non-negative stationary function in $L^1_s$. Our only restriction in this work is that we do not allow point-like charges.

In Section~\ref{sec:pres}, we define the minimization sets and the rHF energy functionals associated with the Yukawa interaction of parameter $m >0$ on the one hand, and with the Coulomb interaction on the other hand. In Section~\ref{sec:EU} we prove the existence of a ground state density matrix $\gamma$, and the uniqueness of the associated ground state density $\rho_\gamma$. We then show in Section~\ref{Y_C} that the $m$-Yukawa rHF model converges to the Coulomb rHF model when the parameter $m$ goes to $0$. Finally, we prove in Section~\ref{sec:SCF} that, in the Yukawa setting, any rHF ground state satisfies a self-consistent equation. 

In Section~\ref{limit_thermo}, we will prove that, still in the Yukawa setting, the rHF model for disordered crystals we have introduced is in fact the thermodynamic limit of the supercell model. 

\subsection{Presentation of the model}
\label{sec:pres}

As in the usual rHF model for perfect crystals~\cite{CDL}, the rHF model we propose consists in minimizing, on the set of admissible density matrices, an energy functional composed of two terms: the kinetic energy per unit volume and the average Coulomb (or Yukawa) energy per unit volume. This leads us to introduce the family of energy functionals
\begin{equation}\label{ee}
\boxed{\E_{\mu,m}(\gamma)=\frac{1}{2}\,\trv\left(-{\Delta} \gamma \right)+\frac{1}{2}D_m(\rho_\gamma-\mu, \rho_\gamma-\mu)}
\end{equation}
with $m=0$ for Coulomb and $m>0$ for Yukawa. The sets of admissible density matrices are defined by 
\begin{equation}
 \K_{\mu, Y} = \left\{ \gamma\in\Sv_{1,1}\cap \underline{\mathcal S },\;  0\leq \gamma\leq 1\; \mbox{a.s.}, 
\;\trv(\gamma)= \EE\bra{\int_Q\mu},\; \rho_\gamma-\mu \in \D_{Y} \right\}
\label{eq:def_K_mu} 
\end{equation}
in the Yukawa setting, and by
$$
 \K_{\mu, C} = \left\{ \gamma\in\Sv_{1,1}\cap \underline{\mathcal S },\;  0\leq \gamma\leq 1\; \mbox{a.s.}, 
\;\trv(\gamma)= \EE\bra{\int_Q\mu},\; \rho_\gamma-\mu \in \D_{C} \right\}
$$
in the Coulomb setting. The constraint $\trv(\gamma)=\EE(\int_Q\mu)$ (neutrality condition) must be added in the latter setting since the average Coulomb energy per unit volume of a non globally neutral stationary charge distribution is infinite (recall that in our definition of $D_0$, we have added a jellium background to enforce the neutrality condition). We also impose this constraint in the Yukawa setting for consistency.  
In our model it is not essential that $\mu\geq0$ but we keep this constraint for obvious physical reasons.

\medskip

The following lemma gives sufficient conditions on $\mu\geq0$ for the sets $\K_{\mu, Y}$ and $\K_{\mu, C}$ to be non empty.

\renewcommand{\labelenumi}{(\roman{enumi})}	
\begin{lemma}\label{non_empty}
If $\mu\in \D_Y$, then $\K_{\mu,Y}$ is non empty. If $\mu\geq 0$ satisfies the following conditions
\begin{enumerate}
\item \label{one1}$\mu\in L^3_s(L^1)$, 
\item \label{two2}there exists $\epsilon>0$ such that $\av{p\bra{\o}}\leq q\bra{\o}\bra{\frac{1}{2}-\epsilon} \;\mbox{a.s.,}$
where $q\bra{\o}=\int_Q\mu(\o,x)\,dx$ and $p\bra{\o}=\int_Qx\mu(\o,x)\,dx$,
\end{enumerate} 
then $\K_{\mu,C}$ is non empty.
\end{lemma}

Loosely speaking, the interpretation of the condition $\av{p\bra{\o}}\leq q\bra{\o}\bra{\frac{1}{2}-\epsilon}$ is that the nuclei do not touch the boundary of $Q$ too often.

\begin{proof}
Let $\mu\in \D_Y$ and $\rho:=\EE(\int_Q\mu)$ {a.s.}~and {a.e}. It is clear that there exists a self-adjoint operator  $\gamma\in \Sv_{1,1}$ such that $0\leq \gamma \leq 1$ a.s. and $\rho_\gamma=\rho$. We can take for instance a free electron gas with constant density $\rho$, that is,
$$\gamma=\1_{(-\ii,\epsilon]}\left(-\Delta\right),\qquad \text{with}\quad \epsilon=\left(\frac{d(2\pi)^d}{|S^{d-1}|}\right)^{2/d}\rho^{2/d}.$$
This state is obviously ergodic since it is fully translation-invariant. Moreover it satisfies
$$\trv\bra{-\Delta\gamma}=\frac{d}{d+2}\left(\frac{d(2\pi)^d}{|S^{d-1}|}\right)^{2/d}\,\rho^{1+2/d},\qquad $$
Besides, $\rho-\mu\in \D_Y$ and therefore $\gamma\in \K_{\mu,Y}$. 

Suppose now that $\mu$ satisfies conditions (i) and (ii) of the statement. Let $\rho$ be the stationary function defined on $Q$ by
\begin{equation}
\rho(\o,x)=\begin{cases}
0&\text{if $q\bra{\o}=0$}\\
\frac{q\bra{\o}}{d(\o)^d}\,\chi\left(\frac{x-\frac{p\bra{\o}}{q\bra{\o}}}{d(\o)}\right)^2&\text{otherwise.}
\end{cases}
\end{equation}
Here $d(\o)={\rm dist}(p\bra{\o}/{q\bra{\o}},\partial Q)$ and  $\chi$ is any non-negative radial function of $C^\infty_c(\RRd)$ with support in $B(0, {1}/{2})$, such that $\int_Q\chi^2=1$. We check that $\rho\in L^2_s\bra{L^q}\cap L^3_s$, where $q$ satisfy the conditions in Proposition~\ref{example3}, and $\sqrt{\rho}\in H^1_s$. Therefore, by the representability Theorem~\ref{n-repre}, there exists a self-adjoint operator  $\gamma\in \Sv_{1,1}$ such that $0\leq \gamma \leq 1$ a.s. and $\rho_\gamma=\rho$. Moreover, \\
$\int_Q\bra{\rho(\o,x)-\mu(\o,x)}\,dx=0$ and $\int_Qx\bra{\rho(\o,x)-\mu(\o,x)}\,dx=0$. It follows from Proposition~\ref{example3} that $\rho-\mu\in \D_C$, and therefore that $\gamma\in \K_{\mu,C}$.
\end{proof}
\renewcommand{\labelenumi}{\arabic{enumi}.}

\subsection{Existence of a ground state}
\label{sec:EU}

Now that we have properly defined the rHF energy, it is natural to look for ground states, that is, minimizers of $\E_{\mu,m}$ on $\K_{\mu,Y/C}$. The ground state energy of a disordered crystal is defined by
\begin{equation}\label{problem_m}
\boxed{I_{\mu,m}=\inf\left\{ \E_{\mu,m}(\gamma),\;  \gamma\in  \K_{\mu,{Y}} \right\}}
\end{equation}
with $m > 0$, in the Yukawa case, and by
\begin{equation}\label{problem_0}
\boxed{I_{\mu,0}=\inf\left\{ \E_{\mu,0}(\gamma),\;  \gamma\in  \K_{\mu,{C}} \right\}}
\end{equation}
in the Coulomb case. 

\renewcommand{\labelenumi}{(\roman{enumi})}
\begin{theorem}[Existence of ergodic ground states]\label{th1_y} Let $0\leq \mu\in L^1_s$.
If $\K_{\mu,Y}$ (resp. $\K_{\mu,C}$) is non empty, then  (\ref{problem_m}) (resp. (\ref{problem_0})) has a minimizer and all the minimizers  share the same density.
\end{theorem}

The proof of Theorem~\ref{th1_y} is based on the weak-compactness of $\K_{\mu,Y/C}$ (Proposition~\ref{prop:weak-compactness}), and on the characterization of the spaces $\D_{C/Y}$ by duality (Corollary~\ref{duality}). We recall that in Lemma~\ref{non_empty} above, we have given natural conditions which guarantee that $\K_{\mu,Y/C}$ is non empty.

\begin{proof}
Let $m\geq 0$ and let $(\gamma_n)$ be a minimizing sequence for $I_{\mu,m}$. As the functional $\E_{\mu,m}$ is the sum of two non-negative terms, these two terms must be uniformly bounded. Since $\trv(-\Delta\gamma_n)$ and $\trv(\gamma_n)=\EE(\int_Q\mu)$ are bounded, we can apply Proposition~\ref{prop:weak-compactness} and extract a subsequence (denoted the same for simplicity), such that $\gamma_n\wto_*\gamma$, with all the convergence properties of the statement of Proposition~\ref{prop:weak-compactness}. In particular, we have
$$\tv{-\Delta\gamma}\leq \liminf_{n\to\ii}\tv{-\Delta\gamma_n} \quad\text{and}\quad \trv(\gamma)=\EE\bra{\int_Q\mu}.$$
Similarly, we know that $z_n:=W_m\ast (\rho_{\gamma_n}-\mu)=(-{\Delta}_s+m^2)^{-1/2}(\rho_{\gamma_n}-\mu)$ is a bounded sequence in $L^2_s$. Thus we can extract another subsequence such that $z_n\wto z$ weakly in $L^2_s$.

Passing to weak limits using that $\rho_{\gamma_n}\wto \rho_\gamma$ in $L^{1+2/d}_s$, it is readily checked that for any $\Phi\in E_{C/Y}$
$$
\left\langle \rho_\gamma-\mu, \bra{-\Delta_s+m^2}^{-\frac{1}{2}}\Phi  \right\rangle_{E_{C/Y}^*,E_{C/Y}} =\left\langle z, \Phi \right\rangle_{L^2_s}. 
$$
Therefore, using the lower semi-continuity of the $L^2_s$-norm, we obtain
$$
\norm{\bra{-\Delta_s+m^2}^{-\frac{1}{2}}\!\bra{\rho_\gamma-\mu}}_{E_{C/Y}^*}= \norm{z}_{L^2_s}\leq \liminf_{\cn} \norm{\bra{-\Delta_s+m^2}^{-\frac{1}{2}}\!\bra{\rho_{\gamma_n}-\mu}}_{L^2_s}.
$$
We deduce from Corollary~\ref{duality} that $\rho_{\gamma}-\mu\in \D_{C/Y}$ and that 
$$
\E_{\mu,m}(\gamma)\leq \liminf_\cn \E_{\mu,m}(\gamma_n)=I_{\mu,m}.
$$
Thus, $\gamma$ is a minimizer of~\eqref{problem_m} (resp.~\eqref{problem_0}). 

Let us now prove the uniqueness of the minimizing density $\rho_\gamma$. Assume that $\gamma_1$ and $\gamma_2$ are two minimizers of ~\eqref{problem_m}  (resp.~\eqref{problem_0}).
A simple calculation shows that
\begin{eqnarray*}
\E_{\mu,m}\left(\frac{\gamma_1+\gamma_2}{2}\right)&=&\frac{1}{2}\E_{\mu,m}\left(\gamma_1\right)+\frac{1}{2}\E_{\mu,m}\left(\gamma_2\right)- \frac{1}{4}D_m\left(\rho_{\gamma_1}-\rho_{\gamma_2},\rho_{\gamma_1}-\rho_{\gamma_2}\right)\\
&=& I_{\mu,m}- \frac{1}{4}D_m\left(\rho_{\gamma_1}-\rho_{\gamma_2},\rho_{\gamma_1}-\rho_{\gamma_2}\right).
\end{eqnarray*}
As $I_{\mu,m}$ is the infimum of $\E_{\mu,m}$ and as $(\gamma_1+\gamma_2)/{2}$ belongs to the minimization set $\K_{\mu,C/Y}$, we deduce that $\|(-{\Delta}_s+m^2)^{-\frac{1}{2}}\left(\rho_{\gamma_1}-\rho_{\gamma_2}\right)\|_{L^2_s}=0$. Thus $(-{\Delta}_s+m^2)^{-\frac{1}{2}}\left(\rho_{\gamma_1}-\rho_{\gamma_2}\right)=0 $. For all $\Phi\in E_{C}$,  $(-{\Delta}_s+m^2)^{1/2}\Phi\in E_{C} $ and 
$$
\left\langle \left(-{\Delta}_s+m^2\right)^{-\frac{1}{2}}\left(\rho_{\gamma_1}-\rho_{\gamma_2}\right), (-{\Delta}_s+m^2)^{\frac{1}{2}}\Phi\right\rangle_{L^2_s} =0.
$$ 
Hence, $\EE(\int_Q(\rho_{\gamma_1}-\rho_{\gamma_2})\Phi)=0$ for all $\Phi\in E_{C}$. As  $E_C$ is dense in $L^{1+d/2}_s\cap\{1\}^\perp$ (see Lemma~\ref{density_result}) and as, in addition, $\EE(\int_Q(\rho_{\gamma_1}-\rho_{\gamma_2}))=0 $, we conclude that $\rho_{\gamma_1}=\rho_{\gamma_2}$.
\end{proof}

\subsection{From Yukawa to Coulomb}\label{Y_C}
In this section, we prove that the ground state energy of the Yukawa problem converges to the ground state energy of the Coulomb problem as the parameter $m$ goes to $0$. The result essentially follows from our definition of the Coulomb energy $D_0$ as the limit of $D_m$ when $m\to0$.

\begin{theorem}[Convergence of Yukawa to Coulomb]\label{th2_y}
Let $0\leq\mu\in L^1_s$ be such that $\K_{\mu,C}\neq\emptyset$.
The function $m\mapsto I_{\mu,m}$ is decreasing and continuous on $[0,+\ii)$. In particular, we have
$$
\lim_{m\to0^+}I_{\mu,m}=I_{\mu,0}.
$$
Moreover, if for each $m > 0$, $\gamma_m$ is a minimizer of (\ref{problem_m}), then the family $\bra{\gamma_m}_{m>0}$ converges, up to extraction, to some minimizer $\gamma_0$ of (\ref{problem_0}), in the same fashion as in Proposition~\ref{prop:weak-compactness}.
\end{theorem}

\begin{proof}
That $m \mapsto I_{\mu,m}$ is decreasing and continuous on $(0,+\infty)$ is easy to check (the strict monotonicity follows from the existence of minimizers). For $f\in \D_C$ such that $\EE(\int_Q f)=0$, we have
$D_m(f,f)\leq D_0(f,f)$ for all $m\geq0$. It follows that  
$$
\forall\gamma\in \K_{\mu,C},\ \forall m > 0,\qquad \E_{\mu,m}(\gamma)\leq  \E_{\mu,0}(\gamma)<\ii ,\;,
$$
and therefore that 
\begin{equation}\label{inequality}
I_{\mu,m}\leq I_{\mu,0}.
\end{equation}
This proves that $\lim_{m\to 0^+}I_{\mu,m}\leq I_{\mu,0}$.

For $m > 0$, we denote by $\gamma_m$ a minimizer of ~\eqref{problem_m}. We deduce from~\eqref{inequality} that there exists a positive constant $C$ such that, for all $m>0$, $\trv\left(-\Delta \gamma_m\right)\leq C$ and $\left\| (-\Delta_s+m^2)^{-1/2}(\rho_{\gamma_m}-\mu)  \right\|_{L^2_s}\leq C$. Reasoning as in the proof of Theorem~\ref{th1_y}, we can extract a subsequence $(\gamma_{m_k})_{k \in \NN}$ with $m_k\searrow0$, such that there exists $\gamma \in \K$ with 
$$
\tv{\gamma}=\|\mu\|_{L^1_s}, \qquad \trv(-\Delta\gamma)\leq\liminf_{k \to \infty}\trv\left( -\Delta\gamma_{m_k}\right),
$$ 
and
$$
\left\| (-\Delta_s)^{-\frac{1}{2}}\left(\rho_{\gamma}-\mu\right)  \right\|_{L^2_s}\leq \liminf_{k \to \infty}  \left\| (-\Delta_s+m_k^2)^{-\frac{1}{2}}\left(\rho_{\gamma_{m_k}}-\mu\right)  \right\|_{L^2_s}. 
$$
This proves that $\gamma \in \K_{\mu,C}$ and that
$$
I_{\mu,0}\leq\E_{\mu,0}(\gamma)\leq \liminf_{k \to \infty} \E_{\mu,m_k}(\gamma_{m_k})=\lim_{m \to 0}  I_{\mu,m} \le I_{\mu,0},
$$
which concludes the proof of the theorem.
\end{proof}


\subsection{Self-consistent field equation}
\label{sec:SCF}
In this section, we define the mean-field Hamiltonian $H=-\frac12\Delta+V$ associated with the ground state for $m>0$ (Yukawa interaction), and we prove that any ground state of~\eqref{problem_m} satisfies a self-consistent field equation. The same holds formally in the Coulomb case but, unfortunately, we are not able to give a rigorous meaning to the Coulomb potential $V$. For this reason we consider a fixed parameter $m>0$ in the rest of the section.
 
We introduce the stationary mean-field potential $V$ defined by
\begin{equation}
V(\o,x)=\int_\RRd Y_m(x-y)\bra{\rho_{\gamma_m}-\mu}(\o, y)\,dy,
\label{eq:def_V_m}
\end{equation}
where $\rho_{\gamma_m}$ is the common density of the minimizers of~\eqref{problem_m}. The following says that, under the appropriate assumptions on $\mu$, $V$ is a well-defined stationary function such that the associated random Schrödinger operator $H=-\frac12\Delta+V$ is also well defined.

\begin{lemma}[Mean-field random Schrödinger operator]
Let $d\in\{1,2,3\}$, $m>0$ and $0\leq\mu \in L^{1+2/d}_s\cap \D_Y$. Let $\rho_{\gamma_m}$ be the (unique) ground state electronic density for the Yukawa minimization problem~\eqref{problem_m}, obtained in Theorem~\ref{th1_y}, and $V$ the associated mean-field potential defined in~\eqref{eq:def_V_m}. Then we have
\begin{equation}\label{eq:V}
V\in \begin{cases}
L^3_s(L^\ii) &\text{ for $d=1$,}\\
L^2_s(L^\ii) &\text{ for $d=2$,}\\
L^{5/3}_s(L^\ii)\cap L^2_s(L^6)&\text{ for $d=3$},
\end{cases}
\end{equation}
and the random Schrödinger operator
$H:=-\frac12{\Delta}+V$
is almost surely essentially self-adjoint on $C^\ii_c(\R^d)$.
In dimension $d=3$, if $\mu\in L_s^{5/2}\bra{L^1}$, then we also have 
\begin{equation}\label{eq:V-inL52}
V_-\in L^{5/2}_s. 
\end{equation}
\end{lemma}

Let us emphasize that $H$ is a uniquely defined operator since $\rho_{\gamma_m}$ is itself unique.
Note that under the sole assumption that $\mu\in L^{1+2/d}_s$ in dimensions $d=1,2$ we have $\mu\in\D_Y$ by Corollary~\ref{cor:in_D_Y}. In dimension $d=3$, the additional hypothesis $\mu\in L^{5/2}_s(L^1)$ ensures that $\mu\in \D_Y$, by Corollary~\ref{cor:in_D_Y} and the fact that $L_s^{5/2}\bra{L^1}\cap L_s^{5/3}\subset L_s^{2}\bra{L^{6/5}}$. 

\begin{proof}
As we know that $\rho_{\gamma_m}\in L^{1+2/d}_s$,~\eqref{eq:V} and~\eqref{eq:V-inL52} follow from Lemma~\ref{lem:convolution} and the fact that $V= c\,W_m\ast(W_m\ast(\rho_{\gamma_m}-\mu))$ with $W_m\ast(\rho_{\gamma_m}-\mu)\in L^2_s$ since $\rho_{\gamma_m}-\mu\in \D_Y$. We know from~\cite[Proposition V.3.2, p.258]{carmona} that $-\frac12\Delta+V$ is essentially self-adjoint on $C^\ii_c(\R^d)$ when $V\in L^r_s(L^p)$ for some $p>2$ and $r>dp/(2(p-2))$. In our case we can apply this with $(p,r)=(3,3)$ for $d=1$, $(p,r)=(5,2)$ for $d=2$ and $(p,r)=(21,5/3)$ for $d=3$.
\end{proof}

The following now gives the self-consistent equation satisfied by a minimizer $\gamma_m$.

\begin{proposition}[Self-consistent equation]\label{forme_minimiseur}
Let $d\in\{1,2,3\}$, $m>0$ and $0\leq\mu \in L^{1+2/d}_s$. Suppose also that $\mu\in L_s^{5/2}\bra{L^1}$ if $d=3$.
Then there exists $\epsilon_{\rm F}\in\R$, called the \emph{Fermi level}, such that any minimizer $\gamma_m$ of the Yukawa minimization problem~\eqref{problem_m} is of the form 
$$
\gamma_m=\1_{\left(-\infty, \epsilon_{\rm F}\right)}(H)+\delta,
$$
for some ergodic self-adjoint operator $\delta$ satisfying $0\leq \delta\leq \1_{\left\{ \epsilon_{\rm F}\right\}}\left( H\right)$.
\end{proposition}

Since $H$ is uniquely defined, we deduce that two different minimizers need to have different operators $\delta$'s at the Fermi level $\epsilon_{\rm F}$. In particular, when $\epsilon_{\rm F}$ is not an eigenvalue of $H$, we deduce that $\gamma_m=\1_{(-\ii,\epsilon_{\rm F})}(H)$ is the unique minimizer of~\eqref{problem_m}.

\begin{proof}
As $\mu\in \D_Y$,~\eqref{problem_m} has a minimizer $\gamma$ by Theorem~\ref{th1_y}. 
The Euler inequality associated with the convex optimization problem~\eqref{problem_m} then reads:
$$
\forall \gamma'\in \K_{\mu,Y},\qquad  \frac12\tv{-\Delta(\gamma'-\gamma)}+D_m(\rho_{\gamma'}-\rho_{\gamma},\rho_\gamma-\mu) \geq0.
$$
For $q \in \R_+$, we set 
$$
E\bra{q}=\inf_{\substack{\gamma'\in\K \atop \tv{\gamma'}=q}\atop \rho_{\gamma'}\in\D_Y}\left(\frac12\tv{-\Delta(\gamma'-\gamma)}+D_m(\rho_{\gamma'}-\rho_{\gamma},\rho_\gamma-\mu)\right).
$$
It is easily checked that the function $E$ is convex on $\R_+$, hence left and right differentiable everywhere. Also, for any 
\begin{equation}\label{def_eps}
\epsilon_{\rm F}\in \left[ E'\bra{\EE\bra{\int_Q\mu}-0},E' \bra{\EE\bra{\int_Q\mu}+0}\right], 
\end{equation}
where $E'(\EE(\int_Q\mu)-0)$ and $E'(\EE(\int_Q\mu)+0)$ respectively denote the left limit and the right limit of the non-decreasing function $E'$ at $\EE(\int_Q\mu)$, we have
$$
\frac12\tv{-\Delta(\gamma'-\gamma)}+D_m(\rho_{\gamma'}-\rho_{\gamma},\rho_\gamma-\mu) - \epsilon_{\rm F}\trv(\gamma'-\gamma) \geq 0
$$
for any ergodic operator $\gamma'\in\K$ such that $\rho_{\gamma'}\in\D_Y$. As $\rho_\gamma\in \D_Y$, $V_\mu=Y_m*\mu\in L^{1+d/2}_s$, and $\rho_{\gamma'}\in L^{1+2/d}_s$ for any $\gamma'\in \K$, the above inequality actually holds for any $\gamma'\in \K$. In addition,
$$D_m(\rho_{\gamma'}-\rho_{\gamma},\rho_\gamma-\mu)=\EE\bra{\int_QV(\rho_{\gamma'}-\rho_\gamma)}$$
in $\RR_+\cup\set{+\ii}$. Taking now $\gamma'=\1_{(-\ii,\epsilon_{\rm F})}(H)$, which belongs to $\K$ by Proposition~\ref{proj}, and using Proposition~\ref{prop:var_charac}, leads to
\begin{multline*}
0\leq \frac12\tv{-\Delta(\gamma'-\gamma)}+D_m(\rho_{\gamma'}-\rho_{\gamma},\rho_\gamma-\mu)-\epsilon_{\rm F}\trv(\gamma'-\gamma)\\ 
\qquad\leq -\tv{|H-\epsilon_{\rm F}|^{1/2}(\gamma'-\gamma)^2|H-\epsilon_{\rm F}|^{1/2}}\leq0.
\end{multline*}
Hence, $\gamma=\gamma'+\delta$ with $\delta$ as in the statement. 
\end{proof}

The following result is a consequence of Proposition~\ref{forme_minimiseur} and of the Feynman-Kac formula. 

\begin{corollary}\label{corollary}
If $\mu\in L^\ii\bra{\O\times \RRd}$, then, for each $m>0$, the common density $\rho_m$ of the minimizers of the Yukawa minimization problem~\eqref{problem_m} is in $L^\ii\bra{\O\times \RRd}$.
\end{corollary}

 
\section{Thermodynamic limit in the Yukawa case}\label{limit_thermo}

The purpose of this section is to provide a mathematical justification of the Yukawa model~\eqref{problem_m} by means of a thermodynamic limit. So far, we did not manage to extend the results below to the Coulomb case.

Let us quickly recall that the thermodynamic limit problem consists in studying the behavior of the energy per unit volume 
(as well as, possibly, the ground state itself and some other properties like the mean-field potential, etc) when the system is confined to a box with chosen boundary conditions and when the size of the box is increased towards infinity.

For a perfect (unperturbed) crystal, the existence of the limit in the many-body case goes back to Fefferman~\cite{Fefferman-85}, after the fundamental work of Lieb and Lebowitz~\cite{LieLeb-72}. A new proof of this recently appeared in~\cite{HaiLewSol_2-09}. However, for the many-body Schrödinger equation, the value of the limiting energy per unit volume is unknown. For effective theories like of Thomas-Fermi or Hartree-Fock type, it is often possible to identify the limit and to prove the convergence of ground states. In~\cite{LiebSimon}, Lieb and Simon prove that, for the Thomas-Fermi model, the
energy per unit volume and the ground state density of a perfect crystal
are obtained by solving a certain periodic Thomas-Fermi model on the unit
cell of the crystal. The same conclusion has been reached by Catto, Le
Bris and Lions for the Thomas-Fermi-von Weizs\"acker model \cite{CLL_book}, and
for the reduced Hartree-Fock (rHF) model~\cite{CLL_periodic} we focus on in the
present work. 

In the stochastic case, Veniaminov has initiated in~\cite{Veniaminov} the study of the thermodynamic limit of random quantum systems, but with short range interactions. The case of a random Coulomb crystal was recently tackled by Blanc and the third author of this article in~\cite{BlaLew-12}. 
Blanc, Le Bris and Lions had already considered Thomas-Fermi like models in~\cite{BLBL2007}, for which they could also identify the limit.

\medskip

Here we follow~\cite{CDL} and we consider the so-called \emph{supercell model}. We put the system in a box $\Gamma_L=\left[-{L}/{2}, {L}/{2} \right)^d$ of side $L\in\NN\setminus\set{0}$, with periodic boundary conditions. When $m>0$, we show that the ground states converge, when $L$ goes to infinity, to a ground state of problem~\eqref{problem_m} (up to extraction and in a sense that will be made precise later).

Let $m>0$ be fixed for the rest of the section. We introduce the Hilbert space 
$$
L^2_{\rm per}\bra{\Gamma_L}=\set{ \phi\in L^2_{\rm loc}\bra{\RRd}\;| \;\phi\; \bra{L\ZZ}^d\mbox{-periodic} }.
$$
The Fourier coefficients of a function $f\in L^2_{\rm per}\bra{\Gamma_L}$ are defined by
$$
c_K^L\bra{f}=\frac{1}{L^\frac{d}{2}}\int_{\Gamma_L}f\bra{x}e^{-iK\cdot x}\,dx,\qquad\forall K\in \left(\frac{2\pi}{L}\ZZ\right)^d.
$$
We denote by $-\Delta_L$ and $P_{j,L}$, $1\leq j\leq d$, the self-adjoint operators on $L^2_{\rm per}\bra{\Gamma_L}$ defined by 
$$
c_K^L\bra{-\Delta_L f}=\av{K}^2c_K^L\bra{f}, \quad\mbox{and}\quad
c_K^L\bra{P_{j,L} f}=k_jc_K^L\bra{f}, \quad \forall K\in \left(\frac{2\pi}{L}\ZZ\right)^d.
$$
For $k\in\ZZd$, we denote as before by $U_k$ the translation operators on $L^2_{\rm loc}\bra{\RRd}$ defined by
$U_kf\bra{x}=f\bra{x+k}$.
For any $f,g\in L^2_{\rm per}\bra{\Gamma_L}$, we set
\begin{align}\label{DL}
D_{m,L}(f,g)&= \av{S^{d-1}}\left\langle\bra{-{\Delta}_L+m^2}^{-\frac{1}{2}}f, \bra{-{\Delta}_L+m^2}^{-\frac{1}{2}}g\right\rangle_{L^2_{\rm per}\bra{\Gamma_L}}\\
&= \sum_{K\in \left(\frac{2\pi}{L}\ZZ\right)^d }\frac{\av{S^{d-1}}}{\left|K\right|^2+m^2}\overline{c_K^L(f)}c_K^L(g)=\int_{\Gamma_L}\int_{\RRd}Y_{m}(x-y)f(x)g(y)\,dx\,dy\nonumber.
\end{align}
Denoting by $\S_{1,L}$ (resp. ${\mathcal S}_L$) the space of the trace class (resp. bounded self-adjoint) operators on $L^2_{\rm per}(\Gamma_L)$, the set of admissible electronic states for the supercell model is then
$$
 \K_L = \left\{ \gamma_L\in\S_{1,L} \cap {\mathcal S}_L,\;  0\leq \gamma\leq 1,\; \tr_{L^2_{\rm per}\bra{\Gamma_L}}\bra{-{\Delta_L} \gamma_L}< \infty \right\}.
$$
For any $\o\in\O$, we denote by $ \mu_L(\o,\cdot) $ the $\bra{L\ZZ}^d$-periodic nuclear distribution which is equal to $\mu(\o,\cdot)$ on $\Gamma_L$, and by $\E_{\mu,m}^L$ the ($\omega$-dependent) energy functional defined on $\K_L$ by 
$$
\E_{\mu,m}^L(\omega,\gamma_L)=\frac{1}{2}\tr_{L^2_{\rm per}\bra{\Gamma_L}}\bra{-{\Delta_L} \gamma_L}+\frac{1}{2}D_{m,L}\big(\rho_{\gamma_L}-\mu_L(\o,\cdot)\,,\,\rho_{\gamma_L}-\mu_L(\o,\cdot)\big).
$$
Let $\epsilon_{\rm F}$ be as in Proposition~\ref{forme_minimiseur}. For any $\omega \in \Omega$, the 
ground state energy of the system in the box of size $L$ with Fermi level $\epsilon_{\rm F}$ is given by 
\begin{equation}\label{finite_problem}
\boxed{I_{\mu,m,\epsilon_{\rm F}}^L(\omega)=\inf\left\{ \E_{\mu,m}^L(\omega,\gamma_L)-\epsilon_{\rm F}\tr_{L^2_{\rm per}\bra{\Gamma_L}}\bra{\gamma_L},\; \gamma_L\in \K_L\right\}.}
\end{equation}

\begin{proposition}
Let $\mu\in L^2_s$. For each $L\in\NN\setminus\set{0}$,~\eqref{finite_problem} has a minimizer, and all the minimizers of~\eqref{finite_problem} share the same density. 
\end{proposition}

\begin{proof}
 The proof follows the same lines as the proof of \cite[Theorem 2.1]{CLL_periodic}, replacing the periodic Coulomb kernel by the periodic Yukawa  kernel $Y_{m,L}\bra{x}=\sum_{k\in \bra{L\ZZ}^d}Y_m(x-k)$. 
\end{proof}

On the other hand, the ground state energy of the full space ergodic problem with Fermi level $\epsilon_{\rm F}$ is defined by
\begin{equation}\label{problem_cp}
I_{\mu,m,\epsilon_{\rm F}}=\inf\left\{ \E_{\mu,m}(\gamma)-\epsilon_{\rm F}\tv{\gamma},\;  \gamma\in \K_Y\right\}, 
\end{equation}
where $\E_{\mu,m}$ is given by~\eqref{ee}, 
\begin{equation}
\K_Y:= \left\{ \gamma\in\Sv_{1,1}\cap \underline{\mathcal S },\;  0\leq \gamma\leq 1\; \mbox{a.s.},\; \rho_\gamma-\mu \in \D_{Y} \right\}
\label{eq:def_K_Y}
\end{equation}
(the neutrality constraint has been removed compared to $\K_{\mu,Y}$ defined before in~\eqref{eq:def_K_mu}). It is a classical result of convex optimization that \eqref{problem_m} and \eqref{problem_cp} have the same minimizers.

\begin{theorem}[Thermodynamic limit for $m>0$]\label{limite_thermo_m}
Let $\mu\in L^2_s$. We have
$$
\lim_{L\to\ii}\frac{I_{\mu,m,\epsilon_{\rm F}}^L(\omega)}{L^d}\,=\, I_{\mu,m,\epsilon_{\rm F}} \quad \mbox{in } L^1\bra{\O}.
$$
\end{theorem}

To prove Theorem~\ref{limite_thermo_m}, we first establish preliminary estimates in Proposition~\ref{pre_est}. Then, we prove a lower bound in expectation in Proposition~\ref{liminf}, and an almost sure upper bound in Proposition~\ref{limsup}. We then conclude the proof of Theorem~\ref{limite_thermo_m} using Lemma~\ref{pm}.

In order to adapt our proof to the Coulomb case, we would need some estimates on the Coulomb potential $V_L$ in the box $\Gamma_L$. It is reasonable to believe that screening effects will make $(V_L)$ bounded in, say, $L^1(\Omega,L^1_{\rm unif}(\R^d))$. For a very general arrangement of the nuclei, bounds of this type are known in Thomas-Fermi theory (see~\cite[Theorem~7]{BLBL2007}, which is taken from Brezis' paper~\cite{Brezis-84}) and in Thomas-Fermi-von Weizsäcker theory~\cite[Theorem~6.10]{CLL_book}, but they have not yet been established in reduced Hartree-Fock theory. Proving such bounds is of considerable interest, but it is beyond the scope of this paper.

\begin{proposition}[Upper bounds]\label{pre_est}
Let $\mu\in L^2_s$ and let $\gamma_L(\o)$ be a minimizer of $I_{\mu,m,\epsilon_{\rm F}}^L(\omega)$. Then, there exists $C>0$ and a sequence of integrable random variables  $(Z_L)$ converging to some $Z\in L^1\bra{\O}$ a.s. and in $L^1\bra{\O}$, such that 
\begin{equation}\label{bound0} I_{\mu,m,\epsilon_{\rm F}}^L(\omega)+D_{m,L}\bra{\rho_{\gamma_L}\bra{\o,\cdot},\rho_{\gamma_L}\bra{\o,\cdot}}\leq C\, L^dZ_L\bra{\o}\;  \mbox{a.s.,}\end{equation}
\begin{equation}\label{bound2}\EE\left(\tr_{L^2_{\rm per}\bra{\Gamma_L}}(1-\Delta_L)\gamma_L\right)+\EE\Big(D_{m,L}\bra{\rho_{\gamma_L}-\mu_L,\rho_{\gamma_L}-\mu_L}\Big)\leq C\,L^d \end{equation}
for all $L\in\NN\setminus\set{0}$.
\end{proposition}

\begin{proof}
Taking $\gamma_L=0$ as a trial state, we obtain that, almost surely,
\begin{equation}
 \frac{I_{\mu,m,\epsilon_{\rm F}}^L(\omega)}{L^d}\leq  \frac{1}{2L^d}D_{m,L}\bra{\mu_L\bra{\o,\cdot},\mu_L\bra{\o,\cdot}}\leq \frac{1}{2m^2} Z_L\bra{\o}, \label{eq:ZLo}
\end{equation}
where $Z_L=L^{-d}\int_{\Gamma_L}\mu^2$ converges to $\EE(\int_Q\mu^2)$, a.s. and in $L^1\bra{\O}$, by the ergodic theorem. 
Besides, we have for any $\alpha\in\RR$ and any $\gamma_L'\in  \K_L$,
\begin{equation}\label{bound5}
\tr_{L^2_{\rm per}\bra{\Gamma_L}}\bra{\bra{-{\Delta}_L-\alpha} \gamma_L'}\geq -\tr_{L^2_{\rm per}\bra{\Gamma_L}}\big(-{\Delta}_L-\alpha\big)_-\geq -CL^d
\end{equation}
where $C$ may depend on $\alpha$ and $d$, but not on $\gamma_L'$.
The bounds~\eqref{bound0} and~\eqref{bound2} follow from~\eqref{eq:ZLo},~\eqref{bound5}, the positivity of each term of $\E^L_{\mu,m}(\omega,\cdot)$, and the fact that $\EE(Z_L)=\|\mu\|_{L^2_s}^2$ is independent of $L$.
\end{proof}

\begin{proposition}[Lower bound in average]\label{liminf}
Let $\mu \in L^2_s$. Then
$$
\liminf_{L\rightarrow +\ii} \frac{\EE\left( I_{\mu,m,\epsilon_{\rm F}}^L(\cdot)\right)}{L^d} \geq I_{\mu,m,\epsilon_{\rm F}}.
$$
\end{proposition}

The following definition introduced in~\cite{BLBL2007} will be used repeatedly in the proof of Proposition~\ref{liminf}.

\begin{definition}
For a function $g:\O\times\RRd\rightarrow \CC$ and $L\in\NN$, we call the \emph{tilde-transform} $\tilde{g}$ of $g$ the following function
\begin{equation}
\tilde{g}\bra{\o,x}=\frac{1}{L^d}\sum_{k\in\Gamma_L\cap \ZZd}g\bra{\tau_{-k}\bra{\o},x+k}\; \mbox{a.s. and a.e.} 
\label{eq:def_tilde_transform}
\end{equation}
\end{definition}

%

We can now write the

\begin{proof}[Proof of Proposition~\ref{liminf}]
Let $\gamma_L(\o)$ be a minimizer of~\eqref{finite_problem} and set 
$$
\tilde{\gamma}_L(\o)=\frac{1}{L^d}\sum_{k\in \Gamma_L\cap \ZZ^d}U_k\gamma_L(\tau_{-k}(\o))U_{k}^*.
$$ 
Notice that $\rho_{\tilde{\gamma}_L}=\widetilde{\rho}_{\gamma_L}$ where the latter is the tilde-transform defined in~\eqref{eq:def_tilde_transform}.
For any $L\in\NN\setminus\set{0}$, we define the operator
$$
\begin{array}{lrll}
          \gamma'_L: &  L^2\bra{\RRd} &\rightarrow & L^2\bra{\RRd}\\
	   & \phi &\mapsto& \1_{\Gamma_L}\tilde{\gamma}_L\phi_L,
           \end{array}
$$
where $\phi_L$ is the $\bra{L\ZZ}^d$-periodic function  equal to $\phi$ on $\Gamma_L$. It is easily checked that  $\gamma'_L$ is self-adjoint and that $0\leq \gamma'_L\leq 1$. Thus, the family $\bra{\gamma'_L}$ is bounded in $L^\ii\bra{\O,\B}$. Up to extraction of a subsequence, there exists an operator $\gamma\in L^\infty\left( \O, \B\right)$ such that $\gamma'_L$ converges weakly-$\ast$ to $\gamma$. Moreover, $\gamma$ is self-adjoint and $0\leq \gamma\leq 1$ a.s. Besides, ${\gamma_L'}\bra{\o,x,y}={\tilde\gamma_L}\bra{\o,x,y}$ a.s. and a.e. on $\O\times \Gamma_L\times \Gamma_L$.
In the following, we will show that $\gamma\in \K_{Y}$ and that $\E_{\mu,m}\bra{\gamma}-\epsilon_{\rm F}\tv{\gamma}\leq \liminf L^{-d}\EE\bra{I^L_{\mu,m,\epsilon_{\rm F}}(\cdot)}$.
\paragraph{Step 1}\textit{The operator $\gamma$ is ergodic.}
Arguing like in the proof of Proposition~\ref{prop:weak-compactness}, it is sufficient to show that for all $u\in L^{1}(\O)$, $\phi,\psi\in C^\ii_c(\RRd)$ and $R\in \ZZd$,
\begin{equation}\label{heart}
\EE\Big(u\big\langle \left(\gamma(\tau_R(\o))-U_R\gamma(\o)U_{R}^*\right)\phi,\psi\big\rangle_{L^2}\Big)=0.
\end{equation}
Let $u, \phi,\psi$ and $R$ as above and $L\in\NN$. We have
$$\tilde{\gamma}_L(\tau_R(\o))-U_R\tilde{\gamma}_L(\o)U_R^*=\frac{1}{L^d}\sum_{k\in \bra{\Gamma_L{\Delta} (\Gamma_L+R)}\cap \ZZd }U_k\gamma_L(\tau_{-k+R}(\o))U_{k}^*,
$$
where $A{\Delta} B:=\bra{A\setminus B} \cup \bra {B\setminus  A}$. Hence, for $L$ sufficiently large, we have
\begin{align*}\label{eq_1}
\left|\EE\bra{u\langle \bra{{\gamma}'_L(\tau_R\cdot)-U_R{\gamma}'_LU_R^*}\phi,\psi\rangle_{L^2}}\right|& =  \left|\EE\bra{u\langle \bra{\tilde{\gamma}_L(\tau_R\cdot)-U_R\tilde{\gamma}_LU_R^*}\phi_L,\psi_L\rangle_{L^2\bra{\Gamma_L}}}\right|\\
&\leq \frac{\left|\Gamma_L{\Delta} (\Gamma_L+R)\right|}{ L^d}\left\| u\right\|_{L^1(\O)}\left\| \phi\right\|_{L^2}\left\| \psi\right\|_{L^2}.
\end{align*}
The left side converges to $\EE\bra{u\langle \left(\gamma\circ\tau_R-U_R\gamma U_{R}^*\right)\phi,\psi\rangle_{L^2}}$, and  the right side  decays as $L^{-1}$. Thus,~\eqref{heart} is proved.

\paragraph*{Step 2}\textit{We have  \begin{equation}\label{step_trace} \trv(\gamma)=\lim_\cL \frac{\Ev{\tr_{L^2_{\rm per}\bra{\Gamma_L}}\bra{\gamma_L}}}{L^d}.\end{equation}}
Thanks to the estimate~\eqref{bound2}, for any $\chi\in W^{1,\infty}_c(\RRd)$, there exists a constant $C$  such that for all $L\in \NN\setminus\set{0}$, we have
$$
 \left|\EE\left( \tr(\chi\tilde{\gamma}_L\chi) \right) \right|+\sum_{j=1}^d\left|\EE\left( \tr(\chi P_j\tilde{\gamma}_L P_j\chi)\right)  \right|\leq C,
$$
Following the proof of Proposition~\ref{prop:weak-compactness}, we can show that 
\begin{equation}\label{bird}
 \EE\left(u\, \tr\bra{ \chi\gamma\chi}\right)=\lim_{L\rightarrow \ii} \EE\left(u\,\int_\RRd \rho_{\gamma'_L} \chi^2 \right)
\end{equation}
for all $u\in L^\ii(\O)$ and all $\chi\in  L_c^{\infty}(\RRd)$.
Choosing $u=1$ and $\chi=\1_Q$, we get
$$\trv(\gamma)=\lim_{L\rightarrow \ii} \EE\left( \int_Q\rho_{\gamma'_L} \right).$$ Finally, we remark that 
$$
\EE\left( \int_Q\rho_{\gamma'_L} \right)= \EE\left( \int_Q\tilde{\rho}_{\gamma_L} \right)= \frac{1}{L^d}\Ev{\tr_{L^2_{\rm per}\bra{\Gamma_L}}\bra{\gamma_L}},
$$
which concludes the proof of~\eqref{step_trace}.
 
\paragraph*{Step 3}\textit{The sequence $({\rho}_{\tilde{\gamma}_L})$ converges weakly to $\rho_\gamma$ in $L^{1+2/d}(\O, L_{\rm loc}^{1+2/d}(\RRd))$.}
By~\eqref{bird}, we obtain
$$
\lim_{\cL}   \EE\left(u \int_\RRd \rho_{\tilde{\gamma}_L}(\o,x)\chi(x)^2\,dx \right) = \EE\left(u \int_\RRd \rho_{\gamma}(\o,x)\chi(x)^2\,dx \right),
$$
for all $u\in L^\infty(\O)$ and all $\chi\in C^\ii_c\bra{\RRd}$. To proceed as in the proof of Proposition~\ref{prop:weak-compactness}, we only need to show that $(\rho_{\tilde{\gamma}_L})$ is bounded in $L^{1+2/d}\left(\O, L^{1+2/d}\left(B_I\right)\right)$, independently of $L$, for any compact set $B_I=\cup_{k\in I}(Q+k)$, with $I\subset\ZZd$ such that card$\bra{I}<\ii$. This bound now follows from the convexity of the function $x\mapsto x^{1+2/d}$ and from the Lieb-Thirring inequality in a box~\cite{FLLS}
\begin{align}\label{bound_lt}
\EE\left( \int_{B_I}\left| \rho_{\tilde{\gamma}_L} \right|^{\frac{d+2}{d}}\right)&\leq \sum_{k\in I}\frac{1}{L^d}\EE\left( \int_{\Gamma_L+k}\left| \rho_{{\gamma}_L} \right|^{\frac{d+2}{d}}\right)\\
&\leq  C(\#I)\bra{\frac{\EE\left( \tr_{L^2_{\rm per}}\bra{ -{\Delta}_L \gamma_L  }\right)}{L^d}+ \frac{\EE\left( \tr_{L^2_{\rm per}}\bra{ \gamma_L  }\right)}{L^d}}\leq C.\nonumber
\end{align}

\paragraph*{Step 4}\textit{We have}
$$\trv\left(-{\Delta} \gamma\right)\leq  \liminf_{L\rightarrow \ii} \frac{\EE\left( \tr_{L^2_{\rm per}}(-{\Delta}_L \gamma_L) \right)}{L^d}.$$
As $\gamma'_L$ converges weakly-$\ast$ in $L^\ii\bra{\O,\B}$ to $\gamma$, we can argue like in the proof of Proposition~\ref{prop:weak-compactness} and get
\begin{eqnarray*}
\trv(-{\Delta} \gamma) &\leq &\liminf_{L\rightarrow\ii} \sum_{j=1}^d\sum_{n\in\NN}\frac{1}{L^d}\sum_{k\in\Gamma_L\cap \ZZd}\EE\left( \langle U_k^*\phi_{n,L},P_{j,L} {\gamma}_L P_{j,L}U_k^*\phi_{n,L}\rangle_{L^2\bra{Q}}\right)\\
&=& \liminf_{L\rightarrow\ii} \frac{1}{L^d}\sum_{j=1}^d\Ev{\tr_{L^2_{\rm per}\bra{\Gamma_L}}\bra{P_{j,L}{\gamma}_LP_{j,L}}}\\
&= &  \liminf_{L\rightarrow\ii} \frac{1}{L^d} \EE\left( \tr_{L^2_{\rm per}\bra{\Gamma_L}}(-{\Delta}_L \gamma_L) \right),
\end{eqnarray*}
where we have used that the operators $P_{j,L}$ commute with the translations $U_k$. 

\paragraph*{Step 5}\textit{ 
We have
\begin{equation}\label{step5}
D_m(\rho_{\gamma}-\mu,\rho_{\gamma}-\mu)\leq\liminf_{L\rightarrow \ii} \frac{\EE\Big( D_{m,L}(\rho_{\gamma_L}-\mu_L,\rho_{\gamma_L}-\mu_L)\Big)}{L^d}.
\end{equation}}
We denote by $f_L=\rho_{\gamma_L}-\mu_L$ and $f=\rho_\gamma-\mu$. It follows from a simple convexity argument that for all $k\in\ZZd$,
$$
\Ev{\left\| \left(-{\Delta}_L +m^2\right)^{-\frac{1}{2}}\tilde{f}_L\right\|_{L^2\bra{Q+k}}^2}
\leq \frac{1}{L^d} \EE\left( \int_{\Gamma_L} \left|\left(-{\Delta}_L  +m^2\right)^{-\frac{1}{2}}f_L\right|^2 \right).
$$
As $$
\int_{\Gamma_L} \left|\left(-{\Delta}_L +m^2\right)^{-\frac{1}{2}}f_L\right|^2 =\av{S^{d-1}}^{-1}\,D_{m,L}\left(f_L,f_L\right),$$
we obtain that for all $k\in \ZZd$,
$$
\left\| \left(-{\Delta}_L +m^2\right)^{-\frac{1}{2}}\tilde{f}_L\right\|_{L^2(\O\times \bra{Q+k})}^2\leq \frac{\av{S^{d-1}}^{-1}}{L^d} \EE\left( D_{m,L}\left(f_L,f_L\right)\right)\leq{C}.
$$
Therefore, there exists a function $z\in L^2( \O, L^2_{\rm unif}(\RRd))$ such that, up to extraction, $(-{\Delta}_L +m^2)^{-1/2}\tilde{f}_L $ converges weakly to $z$ in $L^2( \O, L^2_{\rm unif}(\RRd))$. By the weak lower semi-continuity of the $L^2$-norm, we have
$$
\EE\bra{\int_Qz^2} \leq \liminf_{L\rightarrow \ii} \EE\bra{\left\| \left(-{\Delta}_L +m^2\right)^{-\frac{1}{2}} \tilde{f}_L \right\|_{L^2( Q)}^2}\leq \liminf_{L\rightarrow \ii}  \frac{\EE\bra{ D_{m,L}(f_L,f_L)}}{\av{S^{d-1}}L^d}.
$$
We are going to show that $z=(-{\Delta}_s +m^2)^{-1/2}\bra{\rho_{\gamma}-\mu}$, which will conclude the proof. To do so, we just need to check that for any $u\in L^{1+d/2}\bra{\O}$ and $\chi\in \mathcal C^\ii_c\bra{\RRd}$,
\begin{equation}\label{0}
\lim_{\cL}\Ev{u\int_\RRd\chi \left(-{\Delta}_L +m^2\right)^{-\frac{1}{2}}\tilde{f}_L}=\Ev{u\int_\RRd \left(\left(-{\Delta} +m^2\right)^{-\frac{1}{2}}\chi\right) f}.
\end{equation}
Let $u$ and $\chi$ be such functions. Reasoning as in Step 1, we notice that the tilde-transform $\tilde{\mu}_L$ converges weakly to $\mu$ in $L^{1+2/d}(\O,L^{1+2/d}_{\rm loc}(\RRd))$. Then, we proceed in two steps. First, we show that 
$$\int_{\RRd}\chi \left(-{\Delta}_L +m^2\right)^{-\frac{1}{2}}\tilde{f}_L=\int_\RRd \eta\tilde{f}_L,$$ 
where $\eta=(-{\Delta} +m^2)^{-1/2}\chi$. Recall that, for any $h\in \mathcal S(\RRd)$, the function defined by $h_L\bra{x}=\sum_{k\in \bra{L\ZZ}^d}h\bra{x-k}$ is in $L^2_{\rm per}\bra{\Gamma_L}$ with $c_K^L\bra{h_L}=({2\pi}/{L})^{d/2}\widehat{h}\bra{K}$. For $L$ sufficiently large, we therefore have
\begin{align*}
& \int_{\RRd}\chi \left(-{\Delta}_L +m^2\right)^{-\frac{1}{2}}\tilde{f}_L\\
&\qquad=\bra{\frac{2\pi}{L}}^\frac{d}{2}\sum_{K\in \frac{2\pi}{L}\ZZd}\frac{\overline{\widehat{\chi}\bra{K}} c_K^L\bra{\tilde{f}_L}}{\sqrt{\av{K}^2+m^2}}=\bra{\frac{2\pi}{L}}^\frac{d}{2} \sum_{K\in \frac{2\pi}{L}\ZZd}\overline{\widehat{\eta}\bra{K}} c_K^L\bra{\tilde{f}_L}\nonumber\\
&\qquad= \sum_{K\in \frac{2\pi}{L}\ZZd}\overline{c_K^L\bra{\eta_L}}c_K^L\bra{\tilde{f}_L}= \int_{\Gamma_L}\bra{ \sum_{k\in \bra{L\ZZ}^d} \eta\bra{x-k}} \tilde{f}_L\bra{x}\,dx= \int_\RRd \eta\tilde{f}_L.
\end{align*}
Next, using the fact that $\eta\in\mathcal S\bra{\RRd}$, the weak convergence of $\tilde{f}_L$ to $f$ in $L^{1+2/d}(\O,L_{\rm loc}^{1+2/d}(\RRd))$, and the bound~\eqref{bound_lt}, we obtain 
$$
\Ev{u\int_\RRd\eta\tilde{f}_L}\cvL \Ev{u\int_\RRd\eta f}.
$$
This concludes the proof of~\eqref{0}, hence of~\eqref{step5}.
\end{proof}

\begin{proposition}[Almost sure upper bound]\label{limsup}
Let $\mu \in L^2_s$. Then, 
\begin{equation}\label{eq.limsup}
\limsup_\cL \frac{I_{\mu,m,\epsilon_{\rm F}}^L(\omega)}{L^d} \leq I_{\mu,m,\epsilon_{\rm F}},\; \mbox{a.s.}
\end{equation}
\end{proposition}
\begin{proof}
We will prove~\eqref{eq.limsup} assuming that $\mu\in L^\ii\bra{\O\times \RRd}$; the generalization is obtained by an $\epsilon/2$ argument using~\eqref{bound0} and~\eqref{bound5}.  
 Let $\gamma$ be a minimizer of~\eqref{problem_cp}. By the ergodic theorem, there exists $\O'\subset\O$, with $\PP\bra{\O'}=1$, such that on $\O'$
$$
\lim_{\cL}\frac{1}{L^d}\int_{\Gamma_L}\rho_\gamma=\Ev{\int_Q\rho_\gamma},
\qquad
\lim_{\cL}\frac{1}{L^d}\int_{\Gamma_L}\rho_{P_j\gamma P_j}=\Ev{\int_Q\rho_{P_j\gamma P_j}},
$$
for any $1\leq j \leq d$, and 
\begin{equation}\label{ergo_3}
\lim_{\cL}\frac{\norm{\bra{-{\Delta}_s+m^2}^{-\frac{1}{2}}\bra{\rho_\gamma-\mu}}_{L^2\bra{\Gamma_L}}^2}{L^d}=\norm{\bra{-{\Delta}_s+m^2}^{-\frac{1}{2}}\bra{\rho_\gamma-\mu}}_{L^2_s}^2.
\end{equation}
Let $\o_0\in \O'$ be fixed for the rest of the proof. Let $0\leq\chi_L\leq1$ be a sequence of localization functions of $C^\infty_c\bra{\RRd}$, which equals 1 on $\Gamma_{L-1}$, has its support in $\Gamma_L$, and satisfies $|\nabla\chi_L|\leq C$.
For $L\in\NN\setminus\set{0}$, we introduce the operators $\gamma_L^0: L^2\bra{\RRd}\rightarrow L^2\bra{\RRd}$ and $\gamma_L:  L^2_{\rm per}\bra{\Gamma_L}\rightarrow  L^2_{\rm per}\bra{\Gamma_L}$, whose kernels are given by 
$$
\gamma_L^0\bra{x,y}=\chi_L\bra{x}\gamma\bra{\o_0,x,y}\chi_L\bra{y}
\;\mbox{and}\;
\gamma_L\bra{x,y}=\sum_{j,k\in \bra{L\ZZ}^d}\gamma_L^0\bra{x+j,y+k}.
$$
Using similar techniques to the ones used in the proof of Proposition~\ref{Lieb_T}, one can show that
\begin{equation}\label{trace}
\lim_{L\to\ii} \frac{1}{L^d}\int_{\Gamma_L}\rho_{\gamma_L}=\Ev{\int_Q\rho_\gamma},
\end{equation} 
and that
\begin{equation}\label{kin_ene}
\lim_\cL\frac{1}{L^d}\tr_{L^2_{\rm per}\bra{\Gamma_L}}\bra{-{\Delta}_L\gamma_L}= \tv{-{\Delta} \gamma}.
\end{equation}
We now turn to the convergence of the potential energy, i.e.
\begin{equation}\label{pot_ene}
\lim_\cL\frac{1}{L^d}D_{m,L}\bra{g_L,g_L}= D_m\bra{f,f},
\end{equation}
where $f=\rho_\gamma-\mu$ and  $g_L=\rho_{\gamma_L}-\mu_L$. We introduce the auxiliary function $f_L$, defined as the $L\ZZ^d$-periodic function equal to $f$ on $\Gamma_L$. We first prove that
$$
\lim_\cL \frac{1}{L^d}\bra{D_{m,L}\bra{g_L,g_L}-D_{m,L}\bra{{f}_L,{f}_L}}= 0.
$$
Indeed, rewriting $g_L$ as $g_L=\chi_{L,{\rm per}}^2f_L+\bra{\chi_{L,{\rm per}}^2-1}\mu_L$, with the definition $\chi_{L,{\rm per}}=\sum_{k\in \bra{L\ZZ}^d}\chi_L\bra{\cdot+k}$, we have
\begin{align*}
\norm{\bra{-{\Delta}_L+m^2}^{-\frac{1}{2}}\bra{g_L-f_L}}_{L^2\bra{\Gamma_L}}^2 &\leq  m^{-2}\norm{\bra{g_L-f_L}}_{L^2\bra{\Gamma_L}}^2\\
&\leq m^{-2}\norm{\bra{\chi_{L,{\rm per}}^2-1}\bra{f_L-\mu_L}}_{L^2\bra{\Gamma_L}}^2 
\end{align*}
which is a $o(L^d)$. Then, we prove that
\begin{equation}\label{pot}
\lim_{L\to\ii}\frac{D_{m,L}\bra{{f}_L,{f}_L}}{L^d}= D_m\bra{f,f},
\end{equation}
To do so, in view of~\eqref{ergo_3} it is sufficient to show that
\begin{equation}\label{exact}
\alpha_L= \frac{\av{S^{d-1}}}{L^d}\bra{\norm{\bra{-{\Delta}_L+m^2}^{-\frac{1}{2}}f_L}_{L^2\bra{\Gamma_L}}^2-\norm{\bra{-{\Delta}_s+m^2}^{-\frac{1}{2}}f}_{L^2\bra{\Gamma_L}}^2}
\end{equation}
tends to zero. This follows from the fact that
$$
\alpha_L=\frac{1}{L^d}\int_{\Gamma_L}\,dx\,f\bra{x}\int_{\RRd\setminus \Gamma_L}\,dy \,Y_m\bra{x-y}\bra{f_L\bra{y}-f\bra{y}},
$$
$f\in L^\ii\bra{\O\times \RRd}$ and $Y_m\in L^1\bra{\RRd}$. This completes the proof of~\eqref{pot_ene}. Combining~\eqref{trace},~\eqref{kin_ene} and~\eqref{pot_ene}, we end up with
$$
\limsup_{L\rightarrow\ii} \frac{I^L_{\mu,m,\epsilon_{\rm F}}(\omega)}{L^d}\leq I_{\mu,m,\epsilon_{\rm F}}
$$
for every $\o_0\in\O'$,  which concludes the proof of Proposition~\ref{limsup}.
\end{proof}
We complete the proof of Theorem~\ref{limite_thermo_m} using Lemma~\ref{pm} below applied to $X_L\bra{\o}=L^{-d}I_{\mu,m,\epsilon_{\rm F}}^L(\omega)$ and the bound~\eqref{bound0}.

\begin{lemma}\label{pm}
Let $\bra{X_n}_{n \in\NN}$ be a sequence of random variables in $L^1\bra{\O}$ and $X\in L^1\bra{\O}$. We assume that there exists a sequence of random  variables $\bra{Z_n}_{n\in\NN}$ converging in $L^1\bra{\O}$ to $Z\in L^1\bra{\O}$ such that 
\begin{itemize}
\item $\dps\liminf_\cn\Ev{X_n}\geq \Ev{X}$
\item $\dps\limsup_\cn X_n \leq X \mbox{ a.s.}$
\item $X_n\leq Z_n \mbox{ a.s.}$
\end{itemize}
Then, $X_n\rightarrow X$ strongly in $L^1\bra{\O}$ as $n\rightarrow \ii$.
\end{lemma}
\begin{proof}
Replacing $X_n$ by $X_n-X$, we can assume without loss of generality that $X=0$. We then write $X_n=(X_n)_+-(X_n)_-$. We first notice that $(X_n)_+\rightarrow 0$ a.s. By the dominated convergence theorem with "moving bound" (see e.g. \cite[Theorem 1.8]{LL}), we conclude that  $(X_n)_+ \rightarrow 0$ in $L^1\bra{\O}$.  By the liminf condition, we have $\limsup_{\cn} \Ev{(X_{n})_- }\leq 0$. As $(X_{n})_-\geq 0$, we conclude that $(X_n)_-\rightarrow 0$ in $L^1\bra{\O}$. Finally,
$\EE\bra{\av{X_n}}=\EE\bra{(X_n)_+}+\EE\bra{(X_n)_-}$ tends to $0$.
\end{proof}


\addcontentsline{toc}{section}{Appendix:~Proof of Theorem~\ref{n-repre}}
\section*{Appendix: Proof of Theorem~\ref{n-repre}}\label{AppendixA}

Here we write the proof of Theorem~\ref{n-repre}. This transposition of Lieb's representability theorem to the ergodic setting claims that for any $\rho$ satisfying $\rho\geq 0$, $\rho\in L^3_s$ and $\sqrt{\rho}\in H_s^1$,  
there exists a self-adjoint operator $\gamma\in \Sv_{1,1}$, such that $0\leq \gamma \leq 1$ and $\rho_\gamma=\rho$.

\begin{proof}[Proof of Theorem~\ref{n-repre}]
We start with the case $d=1$. We consider two functions $\phi_0,\phi_1\in C^\infty_c(\RRd)$ satisfying
\begin{itemize}
\item $\phi_0\geq 0$, $\phi_1\geq 0$,
 \item supp$(\phi_0)\subset \left[-\frac{1}{2},\frac{1}{2}\right]$ and supp$(\phi_1)\subset \left[0,1\right]$,
\item $\sum_{k\in\ZZ} \phi_{k}=2$ where $\phi_{2k}(\cdot)=\phi_0(\cdot-k)$ and $\phi_{2k+1}(\cdot)=\phi_1(\cdot-k)$.
\end{itemize}
We denote by $$\rho_k(\o,x):=\rho(\o,x)\phi_k(x),$$ 
and observe that $\rho=\sum_{k\in\ZZ}\rho_k/2$. Let $N_k(\o)=\int_\RRd\rho_k(\o,x)\,dx$. For each $k \in \ZZ$, we set $\phi_{j,k}=0$ for all $j\in\ZZ$ if $N_k\bra{\o}=0$, and
$$
\phi_{j,k}(\o,x)=\frac{\sqrt{\rho_k(\o,x)}}{\sqrt{N_k(\o)}} \exp\left(\frac{2i\pi j}{N_k(\o)}\int_{-\infty}^{x}\rho_k\left(\o,t \right)dt\right)
$$
otherwise. We then introduce the density matrix
$$
\gamma_k=\sum_{j\in\NN}n_{j,k}\left|\phi_{j,k}\left\rangle \right\langle \phi_{j,k}\right|,
$$
where  $n_{j,k}(\o)=\1_{j\leq N_k(\o)}+(N_k(\o)-[N_k(\o)])\1_{j=[N_k(\o)]+1}$.

Each $\gamma_k$ is in $\K=\set{ \gamma \in \Sv_{1,1}\cap \underline{\mathcal S}\;|\;0\leq \gamma_k\leq 1\;\mbox{a.s.} }$ and $\rho_{\gamma_k}(\o,\cdot)=\rho_k(\o,\cdot)$ a.s. As the supports of the kernels of $\gamma_{k}$ and $\gamma_{k+2l}$ are disjoints for all $k,l \in \ZZ$, the operators  $\gamma_e=\sum_{k\in\ZZ}\gamma_{2k}$ and $\gamma_o=\sum_{k\in\ZZ}\gamma_{2k+1}$ are in $\K$. By convexity, so is $\gamma=\frac{\gamma_e+\gamma_o}{2}$. It is finally easily checked that  $\rho_\gamma=\rho$.

We now turn to the case $d=2$. In the same spirit as for $d=1$, we cover the space with a finite number of periodic patterns, in such a way that the elements of each pattern do not intersect (see Figure~\ref{fig:grille}). For example, let
$$
A_0=\left[-\frac{5}{12},\frac{5}{12}\right)^2\!\!,\;  B_0=\left[ \frac{1}{3} , \frac{2}{3} \right)\times \left[ -\frac{1}{4} ,\frac{1}{4}\right)\cup\left[ -\frac{1}{4} ,\frac{1}{4}\right)\times \left[ \frac{1}{3} , \frac{2}{3} \right)\!, \; C_0=\left[ \frac{1}{6},\frac{5}{6} \right)^2. 
$$ 
The $\ZZ^2$-translations of these sets $I_k=I_0+k$, $I\in\{A,B,C\}$, satisfy 
$I_k\cap I_j=\emptyset$  for $k\neq j$ and $\cup_{k\in\ZZ^2}A_k\cup B_k \cup C_k=\RR^2$.
Next, we consider three sequences of regular functions $(\phi^I_k)_{k\in\ZZd}$, $I\in\{A,B,C\}$, such that 
$$\phi_k^I\geq 0,\quad\mbox{supp}(\phi^I_k)\subset I_k,\quad\mbox{and }\sum_{k\in\ZZ^2}\phi_k^A+ \phi_k^B +\phi_k^C=3.$$
\begin{figure}[h]
\centering
\includegraphics{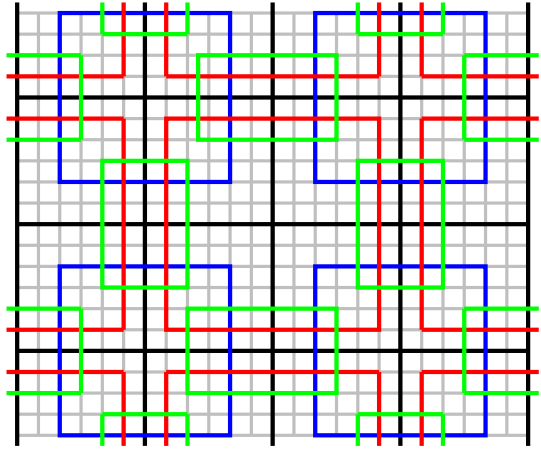}
\caption{Covering in dimension $d=2$ used in the proof of Theorem~\ref{n-repre}.\label{fig:grille}}
\end{figure}

Repeating the argument detailed above in the one-dimensional case, we define $\gamma_I$, for  $I\in\{A,B,C\}$, and $\gamma={\sum_{ I\in\{A,B,C\}}\gamma_I}/{3}$ and we check that $\rho_\gamma=\rho$ and that $\gamma$ satisfies the desired conditions.
We proceed similarly for $d\geq 3$.
\end{proof}

\bibliographystyle{model1b-num-names}
\bibliography{fichierb}

\end{document}